\documentclass[10pt, conference]{IEEEtran}
\ifCLASSINFOpdf
\else
\fi
\usepackage{tikz}
\usepackage{graphicx}
\usetikzlibrary{shapes,arrows,automata}
\usepackage[nointegrals]{wasysym}

\usepackage{mathpartir}
\usepackage{amsmath,amsfonts,amsthm,mathtools,thmtools}
\usepackage{stmaryrd}
\usepackage{thm-restate}
\usepackage{environ}

\usepackage{listings}
\usepackage[inline]{enumitem}

\usepackage{prftree}

\usepackage{xcolor}

\usepackage{fdsymbol}

\usepackage[hyphens]{url}
\usepackage{hyperref}
\usepackage[hyphenbreaks]{breakurl}

\renewcommand{\tablename}{Table}
\renewcommand{\figurename}{Figure}

\newcommand{\ttype}[1]{\mbox{}\hspace{-2mm}{{\small \texttt{ #1}}}}

\lstdefinelanguage{cil}{
    keywords = {allow, auditallow, deny, dontaudit, neverallow, allowx, 
        auditallowx, dontauditx, neverallowx, call, macro, 
        common, classcommon, class, classorder, classpermission,
        classpermissionset, classmap, classmapping, permissionx,
        boolean, booleanif, tunable, tunableif, constrain,
        validatetrans, mlsconstrain, mlsvalidatetrans, block,
        blockabstract, blockinherit, optional, context,
        defaultuser, defaultrole, defaulttype, defaultrange,
        filecon, fsuse, genfscon, ibpkeycon, ibendportcon,
        sensitivity, sensitivityalias, sensitivityaliasactual,
        sensitivityorder, category, categoryalias,
        categoryaliasactual, categoryorder, categoryset,
        sensitivitycategory, level, levelrange, rangetransition,
        ipaddr, nefifcon, nodecon, portcon, mls, handleunknown,
        policycap, reject, role, roletype, roleattribute,
        roleattributeset, roleallow, roletransition, rolebounds,
        sid, sidorder, sidcontext, type, typealias,
        typealiasactual, typeattribute, typeattributeset,
        expandtypeattribute, typebounds, typechange, typemember,
        typetransition, typepermissive, user, userrole,
        userattribute, userattributeset, userlevel, userrange,
        userbounds, userprefix, selinuxuser, selinuxuserdefault,
        iomemcon, ioportcon, pcidevicecon, pirqcon,
        devicetreecon, name, true, false, and, or, xor, eq,
        neq, not, all, dom, domby, incomp, range, r1, r2,
        r3, t1, t2, t3, u1, u2, u3, l1, l2, h1, h2,
        string, dccp, sctp, tcp, udp, task, trans, xattr,
        self, low, high, low_high, glblub, source, target}
}

\lstdefinelanguage{NuSMV}{
    keywords = {DEFINE, MODULE, VAR, IVAR, TRANS, LTLSPEC, next}
}

\lstdefinestyle{cil}{basicstyle=\footnotesize\ttfamily, language=cil}
\lstdefinestyle{NuSMV}{basicstyle=\footnotesize\ttfamily, language=NuSMV}

\newtheorem{corollary}{Corollary}
\newtheorem{lemma}{Lemma}
\newtheorem{definition}{Definition}

\def\makenewenum#1#2{%
\newcounter{cnt#1}
\newenvironment{#1}%
{\begin{list}{\makebox[0pt][r]{#2}}%
{\setlength{\itemsep}{0pt}%
 \setlength{\parsep}{.2em}%
 \setlength{\leftmargin}{2em}%
 \setlength{\labelwidth}{.2em}%
 \usecounter{cnt#1}}}%
{\end{list}}}
\makenewenum{enu}{\arabic{cntenu}.}
\makenewenum{enum}{(\arabic{cntenum})}
\makenewenum{itenum}{\em(\arabic{cntenum})}
\makenewenum{en}{(\roman{cnten})}
\makenewenum{renum}{\textit{(\roman{cntrenum})}}
\makenewenum{enhyphen}{(\roman{cntenhyphen}')}
\makenewenum{aenum}{(\alph{cntaenum})}
\makenewenum{itemizes}{$\bullet$}

\newenvironment{restate-theorem}[1]%
  {\begin{trivlist}\item[]{\normalsize\bfseries{\sffamily}
        Restatement of Theorem~#1.}\hspace*{0mm}\it}%
  {\end{trivlist}}

\newenvironment{restate-lemma}[1]%
  {\begin{trivlist}\item[]{\normalsize\bfseries{\sffamily}
        Restatement of Lemma~#1.}\hspace*{0mm}\it}%
  {\end{trivlist}}

\newenvironment{restate-proposition}[1]%
  {\begin{trivlist}\item[]{\normalsize\bfseries{\sffamily}
        Restatement of Proposition~#1.}\hspace*{0mm}\it}%
  {\end{trivlist}}

\NewEnviron{smalign*}{%

  \scalebox{0.9}{\parbox{1.01\linewidth}{%
  \begin{align*}
    \BODY
  \end{align*}
}}

}

\newcommand\ta{\ensuremath{\mathit{ta}}}

\newcommand{\ltlencode}[1]{\scalebox{1.3}{\leftmoon}\mkern-7mu #1 \mkern-5mu\scalebox{1.3}{\rightmoon}}
\newcommand{\lang}{IFCIL}
\newcommand{\tool}{IFCILverif}
\newcommand{\modelchecker}{NuSMV}
\newcommand{\semantics}[1]{\llbracket#1\rrbracket}

\begin{document}
%
\IEEEoverridecommandlockouts
\title{IFCIL: An Information Flow Configuration Language for SELinux

(Extended Version)}



\author{
\IEEEauthorblockN{
  Lorenzo Ceragioli\IEEEauthorrefmark{1},
  Letterio Galletta\IEEEauthorrefmark{2}\IEEEauthorrefmark{4},
  Pierpaolo Degano\IEEEauthorrefmark{1}\IEEEauthorrefmark{2},
  and David Basin\IEEEauthorrefmark{3}
}
\IEEEauthorblockA{\IEEEauthorrefmark{1} Universit\`a di Pisa, Italy}
\IEEEauthorblockA{\IEEEauthorrefmark{2} IMT School for Advanced Studies Lucca, Italy}
\IEEEauthorblockA{\IEEEauthorrefmark{3} ETH Zurich, Switzerland}
\IEEEauthorblockA{\IEEEauthorrefmark{4} CINI Cybersecurity National Laboratory, Rome, Italy}
}


%

\maketitle

\begin{abstract}
    Security Enhanced Linux (SELinux) is a security architecture for Linux implementing mandatory access control. 
    It has been used in numerous security-critical contexts  ranging from servers to mobile devices.
    But this is challenging as SELinux security policies are difficult to write, understand, and maintain.
    Recently, the intermediate language CIL was introduced to foster the development of high-level policy languages
    and to write structured configurations. 
    However, CIL lacks mechanisms for ensuring that the resulting configurations obey desired information flow policies.
    To remedy this, we propose \lang, a backward compatible extension of CIL for specifying fine-grained information flow requirements for CIL configurations.
    Using \lang, administrators can express, e.g., confidentiality, integrity, and non-interference properties.
    We also provide a tool to statically verify these requirements.
\end{abstract}

\section{Introduction}\label{sec:intro}
%

Security Enhanced Linux (SELinux) is a set of extensions of the Linux kernel that implements a Mandatory Access Control mechanism.
It is widely used for defining security polices in Linux-based systems, including servers~\cite{Yokoyama}, network appliances~\cite{openWRTconfig}, and mobile devices~\cite{SmalleyC13}.
Defining an SELinux policy is conceptually simple: the system administrator defines a set of \emph{types}, uses them to label all system resources and processes, and then defines a set of rules specifying which operations the processes can perform on resources.
However, its use is far from simple.
Writing, understanding, and maintaining SELinux security policies is difficult and error-prone as evidenced by numerous examples of  misconfigurations~\cite{ImCW18} that have led to serious vulnerabilities in widely used policies.

To simplify working with SELinux and to address the limitations of its default policy language, the community called for and proposed new high-level configuration languages~\cite{issueHLL,Lobster}.
In particular, SELinux developers recently proposed the intermediate configuration language CIL (Common Intermediate Language),
which is a declarative language that offers advanced features to aid both policy specification and analysis.
CIL supports the definition of structured configurations, using, e.g., namespaces and macros, 
and enables administrators to specify which resources are critical, which entities can access them, and which cannot.
It also provides tool support to statically detect and prevent misconfigurations, which could lead to unauthorized access to security-critical resources.

However,  CIL currently provides no means to prevent unwanted indirect information flows, which is essential to preventing confidentiality and integrity breaches.
To overcome this serious limitation, we propose \lang, an extension of CIL supporting information flow requirements, and we endow it with a verification procedure for statically checking that a configuration satisfies its requirements.

Our proposal consists of three parts. 
First, we propose the domain specific language (DSL), called \emph{IFL} (Information Flow Language), for expressing fine-grained information flow requirements, which we group in two categories:
\emph{functional} and \emph{security} requirements.
Functional requirements specify which permissions must be granted to users to perform their authorized tasks, such as which resources they can access and with which operations.
In contrast, security requirements prevent entities from operating on other possibly critical entities, and thereby enforce security properties, including confidentiality, integrity, and non-transitive information flow properties. 
Our DSL is compositional and supports the refinement of requirements with further restrictions both to make them more demanding and to adapt them to specific contexts.  

Second, we introduce \emph{\lang} (Information Flow CIL), which extends CIL with constructs to annotate configurations with IFL requirements.
Our extension  is backward compatible: an \lang\ configuration is also a valid CIL configuration and can be translated by the standard CIL compiler.

Finally, we  endow \lang\ with a verification procedure supported by an automated tool that, given a configuration, checks if its IFL requirements are satisfied.
We assess our tool's effectiveness and scalability on real-world configurations.

In summary, our main contributions are as follows.
\begin{itemize}
\item
We present the language IFL for expressing complex, fine-grained, information flow requirements in a declarative and
compositional way, including confidentiality, integrity, and non-transitive information flow properties. IFL requirements can be extended through
refinement and various access control languages can easily be augmented with IFL. Moreover, IFL requirements can be verified using off-the-shelf LTL model checkers.

 \item We propose \lang, the integration of IFL inside CIL. 
	We achieve this by using special comments that an administrator can associate with different parts of a CIL configuration. 
    We give an algorithm for statically verifying the compliance of a configuration to its IFL requirements. 
    
    \item We give CIL a formal semantics and empirically validate its adequacy with respect to the CIL reference manual and the CIL compiler.
	Besides providing the basis for our verification algorithm, the semantics and its experimental validation make it possible to understand CIL's trickier bits, and to illuminate some unspecified corner cases and disagreements between the documentation and the compiler.

    \item We provide a prototype tool~\cite{Tool}
    that implements our verification procedure by leveraging \modelchecker, a popular model checker~\cite{NuSMV2}. 
        Our tool checks if an \lang\ configuration satisfies the requirements therein and, when they are violated, it warns the administrator about potentially dangerous parts of the configuration.
      
    \item 
    We experimentally assess our tool on three real-world CIL policies~\cite{openWRTconfig, cilbase, dspp5}.
    We annotate them with IFL requirements expressing properties taken from the literature and with new ones. 
    We thereby validate our tool and show that it scales well. For example, it takes less than two minutes to verify 39 requirements on the configuration in~\cite{openWRTconfig}, which has roughly 46,000 lines of code.
    
%
\end{itemize}

\noindent
\emph{Outline:}
In Section~\ref{sec:back} we introduce SELinux, CIL, and the mechanism used by administrators to protect critical resources.
In Section~\ref{sec:formal} we give a high-level account of our CIL semantics and how we experimentally validate its adequacy.
In Section~\ref{sec:IFL} we present \lang\ and we explain our verification procedure for checking the satisfiability of the requirements in Section~\ref{sec:verification}.
%
%
In Section~\ref{sec:tool} we present our verification tool and our experimental assessment.
In Section~\ref{sec:related} we compare our work with the relevant literature and in Section~\ref{sec:conclude} we draw conclusions.   
The appendices contain
the details of our formal development and the proofs of our theorems.

\section{Background}\label{sec:back}

\paragraph{SELinux}
SELinux is a set of extensions to the Linux kernel and utilities~\cite{SELinux}.
It extends the major subsystems of the Linux kernel with strong, flexible, mandatory access control (MAC).
The SELinux security server permits or denies a process to invoke a system call on a resource based on a configuration specified by the system administrator.
To specify a configuration, an administrator defines a set of \emph{types}, and labels the OS resources and processes with them.
In addition, all resources belong to predefined \emph{classes}, such as file, process, socket, or directory.
A rule in a configuration relates the type $t$ and class $c$ of resources, and the type $t'$ of processes with the permitted operations.
A rule thereby specifies the actions that processes labelled $t'$ can perform on the resources of class $c$ labeled $t$, for example, read or write a file, execute a process, open a socket, or change the DAC rights of a directory.
A process \emph{P} can invoke a system call \emph{SC} on a resource \emph{R} only if there is a rule that permits \emph{P} to do so.

Administrators typically specify configurations using SELinux's \emph{kernel policy language}~\cite{Kerneldoc}.
Configurations are then compiled to a kind of (kernel binary) access-control matrix. 
However, this policy language is very low-level.
For example, it does not allow the administrator to structure configurations, which makes them hard to understand and maintain.
Thus using the kernel policy language is cumbersome and error-prone, as shown by the over permissive evolution of the Android policy~\cite{ImCW18}.
Some high-level configuration languages have been suggested with their own compilers and tools as an attempt to address these limitations~\cite{Lobster, Nakamura}.
Recently, the SELinux developers proposed a promising new intermediate configuration language with  advanced features and tools to support both the development of high-level languages and the definition of configurations. 
We briefly survey this language below.

\paragraph{CIL}
The Common Intermediate Language (CIL)~\cite{CILdoc} was designed as a bridge between high-level configuration languages and the low-level binary representation introduced above.
Compilers from various configuration languages to CIL are intended to support multi-language policy definitions.
A compiler for the kernel policy language is currently available, and CIL is designed to support existing high-level configuration languages, e.g., Lobster~\cite{Lobster}, and future ones too.
Despite its original goal, CIL is also used to directly write configurations~\cite{dspp5,openWRTconfig,cilbase} for complex real-world policies, like for Android~\cite{aosp}.
Indeed, CIL provides its users with high-level constructs like nested blocks, inheritance, and macros, thereby supporting the structured definition of configurations.
Moreover, since CIL is declarative, it facilitates reasoning about configurations, and the same analysis techniques and tools for CIL can help when other high-level languages are used.

Roughly, a CIL configuration consists of a set of declarations of blocks, types, and rules.
Similarly to classes in programming languages, \emph{blocks} have names and introduce  namespaces and further declarations. 
\emph{Types} are labels that are associated with system resources and processes.
Rules regulate types by specifying which operations processes can perform on resources.
Intuitively, administrators can define two kinds of rules: those that grant permission to processes (\emph{allow} rules) and those that specify permissions that must be never granted to processes (\emph{never allow} rules).

Types can be grouped into named sets, called \emph{typeattributes}, which may be used inside rules to denote all the types therein.
Blocks can also contain \emph{macro} definitions that allow an administrator to abstract a set of rules and to reuse them in different parts of a configuration. 
Macros can have types as parameters that are instantiated when the macro is called.
Moreover, to foster code reuse and modularity, CIL features the construct 
\ttype{blockinherit} 
that permits a block to \emph{inherit} from another block. 
Similarly to Object Oriented languages, all the definitions of rules and types in the inherited block are available in the inheriting block. 
The main difference is that inheritance is actually realized by a kind of copying rule.
%

The most appealing features of CIL with respect to the kernel policy language of SELinux are blocks that enable the administrator 
build modular configurations, as well as macros and inheritance that allow code reuse.

Below, we illustrate CIL's main features through examples.
These examples also illustrate that blocks, types, typeattributes and macros have names, and resolving them in the correct name space and order is non-trivial.

Consider the following CIL block \ttype{house} that declares two types, \ttype{man} and \ttype{object}, and the permission (the \ttype{allow} rule) for processes labeled \ttype{man} to read the files labeled \ttype{object}:
\vspace{-1mm}
\begin{lstlisting}[style=cil]
(block house
    (type man)
    (type object)
    (allow man object (file (read))))
\end{lstlisting}
%
%
%
%
Intuitively, processes of type \ttype{house.man} can read the elements of the class \ttype{file} labeled \ttype{house.object}.
Note that blocks introduce namespaces, and the elements defined therein may be referred to directly within the block itself, or by their qualified name, as done above.

The following block inherits the types \ttype{man} and \ttype{object} and the relevant permission from the block \ttype{house} through the \ttype{blockinherit} rule. 
\begin{lstlisting}[style=cil]
(block cottage
    (blockinherit house)
    (type garden))
\end{lstlisting}
Intuitively, \ttype{blockinherit} copies the body of the block \ttype{house}. Thus the qualified names of the copied types become \ttype{cottage.man} and \ttype{cottage.object}.
In contrast, the type \ttype{garden} is declared in the block, which is not in \ttype{house}.

Blocks can be nested, and the outermost block can refer to the elements in the nested ones by qualifying their names.
\begin{lstlisting}[style=cil]
(block tree
  (block nest
     (type egg))
  (type bird)
  (allow bird nest.egg (file (write))))
\end{lstlisting}
%
%
Intuitively, the last \ttype{allow} rule grants subjects with type \ttype{tree.bird} 
the permission to write to the 
files with type \ttype{tree.nest.egg}.

A global namespace is assumed that includes all the blocks, 
the global types, and the global permission.
For example, in 
\begin{lstlisting}[style=cil]
(type stranger)
(allow stranger inhouse.object (file (open)))
(block inhouse
    (type man)
    (type object)
    (allow man object (file (read)))
    (allow .stranger object (file (read)))
    (allow stranger object (file (write))))
\end{lstlisting}
 the name \ttype{stranger} and the fully qualified \ttype{.stranger} in the allow rules both refer to the global type \ttype{.stranger}.
Note however that if the block \ttype{inhouse} declared a type \ttype{stranger}, this declaration would overshadow the global one in the last \ttype{allow} rule, but not the third one since a fully qualified name is used.
Note too that the global \ttype{allow} rule refers to a type declared in the enclosed block.

The administrator can collect a set of rules using a macro-like construct, as shown in the following example.
\begin{lstlisting}[style=cil]
(block animal_mcr
  (macro add_dog((type x)(type y))
    (allow x man (file (read)))
    (allow y dog (file (open))))
  (type dog))
\end{lstlisting}
Macros are invoked as follows.
\begin{lstlisting}[style=cil]
(block animal_house
  (type man)
  (type cat)
  (call animal_mcr.add_dog(cat cat)))
\end{lstlisting}
Roughly, the content of \ttype{add\_dog} replaces the last line where the formal parameters \ttype{x} and \ttype{y} are bound to \ttype{animal\_house.cat}.
Names are resolved using a mechanism similar to dynamic binding: the name \ttype{dog} in the macro is resolved as \ttype{animal\_mcr.dog}, while \ttype{man} is resolved as \ttype{animal\_house.man}.
Name resolution can be rather intricate, especially when constructs are combined in non-trivial ways, such as when inheritance and macros are interweaved.
In these cases, configurations may have unexepected behaviour 
(see Section~\ref{sec:formal} for examples), 
and lead to misconfigurations that are difficult to spot. 
This problem is exacerbated by the fact that administrators cannot refer to a formal semantics, which CIL lacks.
One contribution of this paper is to provide such a semantics.
We define it in Appendix~\ref{app:CIL} and provide an intuitive account in Section~\ref{sec:formal}. 

An administrator can group types into named sets, called type attributes, which may be used in place of a type. 
The following declares two type attributes named \ttype{pet} and \ttype{not\_pet} and defines the types therein.
\begin{lstlisting}[style=cil]
(typeattribute pet)    
(typeattributeset pet 
    (or (animal_mcr.dog) (animal_house.cat)))
(typeattribute not_pet)
(typeattributeset not_pet
    (not (pet)))
\end{lstlisting}
The first type attribute includes the two types \ttype{animal\_mcr.dog} and \ttype{animal\_house.cat}. 
In contrast, the second one includes all the others.

Administrators can also specify which permissions should never be granted to a given type using \ttype{neverallow} rules.
The rule below prohibits subjects with type
\ttype{animal\_house.cat} to read resources of any type not in \ttype{pet}:
\begin{lstlisting}[style=cil]
(neverallow animal_house.cat not_pet (file(read)))
\end{lstlisting}  
The CIL compiler statically checks that no \ttype{allow} rule inside the configuration violates a \ttype{neverallow} rule.
In this example the compiler will report an error because \ttype{animal\_house.cat} can read the files of type \ttype{animal\_house.man} that is in \ttype{not\_pet}.
Although useful, as we explain below, these checks are insufficient to prevent insecure information flow.

\paragraph{An example from the security domain}

Consider the following block \ttype{mem} defined in~\cite{openWRTconfig},
a CIL configuration designed for OpenWrt powered wireless routers.
\begin{lstlisting}[style=cil]
(block mem
  (block read
    (typeattribute subj_typeattr)
    (typeattribute not_subj_typeattr)
    (typeattributeset not_subj_typeattr 
        (not subj_typeattr))
    (neverallow not_subj_typeattr nodedev 
                       (chr_file (read)))))
\end{lstlisting}
This block defines an inner block \ttype{read} and two disjoint type attributes. 
The first includes the system subjects, and the second includes other types.
The \ttype{neverallow} rule prevents \ttype{not\_subj\_typeattr} types from reading a character file of the globally defined type \ttype{nodedev}.
The underlying idea is that resources of type \ttype{nodedev} are critical for the system and must be carefully protected.
This block shows a typical pattern that administrators use to protect critical resources in CIL using type attributes and \ttype{neverallow} rules. 

This pattern offers an extra check.
In our example, if the administrator includes the following rule 
\begin{lstlisting}[style=cil]
(allow untrusted mem.read.nodedev (chr_file (read)))
\end{lstlisting}
that grants a type \ttype{untrusted} the permission to read a character file of type \ttype{nodedev}, then the CIL compiler raises an error. 
There are two ways to avoid this error: the administrator may either remove the last rule (because granting the permission is actually dangerous), or add \ttype{untrusted} to \ttype{subj\_typeattr} to grant the permission.

However, this pattern is insufficient to control how information flows.
For example, consider the following snippet
\begin{lstlisting}[style=cil]
(type untrusted)
(type vect)
(type deputy)
(typeattributeset mem.read.subj_typeattr deputy)
(allow deputy mem.read.nodedev (chr_file (read)))
(allow deputy vect (file (write)))
(allow untrusted vect (file (read)))
\end{lstlisting}
where the types \ttype{untrusted}, \ttype{vect}, and \ttype{deputy} are defined, and \ttype{deputy} is in \ttype{mem.read.subj\_typeattr}.
Now, a leak may occur if a subject in \ttype{subj\_typeattr} reads a character file of type \ttype{nodedev} and forwards information, via \ttype{vect}, to an arbitrary process of type \ttype{untrusted}, which is permitted by the given \ttype{allow} rules.

\paragraph{Preventing information flow}
Currently, CIL does not prevent indirect information flows between types.
The goal of our work is to extend it with a DSL, dubbed \emph{IFL}, to express information flow control requirements.
We call the resulting language \lang.
In addition, we endow \lang\ with a mechanism for statically checking that a configuration satisfies the stated requirements.
Our extensions provide administrators with an extra, automatic check when defining rules that grant or deny information flows from a critical resource.

We provide some intuition behind our extension by adding the following lines to the \ttype{mem} block above:
\begin{lstlisting}[style=cil]
(typeattribute ind_subj_typeattr)
(typeattribute not_ind_subj_typeattr)
(typeattributeset not_ind_subj_typeattr 
    (not ind_subj_typeattr))
;IFL; ~(nodedev +> not_ind_subj_typeattr) ;IFL; 
\end{lstlisting}
The first three lines introduce two type attributes \ttype{ind\_subj\_typeattr}, and \ttype{not\_ind\_subj\_typeattr}, which are declared disjoint.
The last line, enclosed between the \ttype{;IFL;} markers  is IFL \emph{annotation} that specifies the infomation flow requirement that no information can flow from \ttype{nodedev} to \ttype{not\_ind\_subj\_typeattr}.
This annotation is given as a CIL comment that is used by our verification tool, but is completely ignored by the standard CIL compiler. 
Thus, an \lang\ configuration is still a CIL configuration. 

Note that IFL enables administrators to use a pattern similar to the pattern used with \ttype{neverallow}, preventing \ttype{not\_ind\_subj\_typeattr} types from getting information from a character file of type \ttype{nodedev}.
In this way, our tool warns the administrator of the information leakage from \ttype{nodedev} character files illustrated above.


\section{Formalizing CIL}\label{sec:formal}
%



The official CIL documentation~\cite{CILdoc} does not formally describe CIL's syntax and semantics.
The following, admittedly artificial, configuration highlights the need for a formal semantics:
\begin{lstlisting}[style=cil]
(type a)
(block A
  (call B.m1(a)))
(block B
  (macro m1((type x))
    (type a)
    (allow a x (file (read)))))
\end{lstlisting}
One would expect the parameter \ttype{x} of the macro \ttype{B.m1} to be bound to the type \ttype{a} in the global namespace, thereby allowing \ttype{A.a} to read files of type \ttype{.a}.
Instead, \ttype{x} is bound to \ttype{A.a}, and the resulting permission for \ttype{A.a} is to read files of type \ttype{A.a}.

As a second example, consider the following configuration:
\begin{lstlisting}[style=cil]
(type a)
(macro m((type x))
  (type b)
  (allow x b (file (read))))
(block A
  (call m(a)))
(block B
  (type a)
  (blockinherit A))
\end{lstlisting}
Here the block \ttype{B} inherits from \ttype{A}, which calls the macro \ttype{m}.
There are two plausible orders in which macro calls and inheritances can be resolved, and the choice determines to which name the parameter \ttype{x}  is bound when the \ttype{allow} rule is copied in \ttype{B}.
If the macro call is resolved before inheritance, then \ttype{x} is bound to \ttype{.a} (since \ttype{a} is undefined in \ttype{A}).
If instead the inheritance is resolved first, then the call instruction is copied inside \ttype{B} and \ttype{x} is bound to \ttype{B.a}.
This is CIL's actual behaviour, but the reference guide is unclear about the choice.


\paragraph{Ambiguities in CIL}\label{sec:ambiguity}

We found cases that are counterintuitive, but nevertheless are represented by our semantics correctly, i.e. in accordance with the actual behaviour of the CIL compiler.
For example, the following
\begin{lstlisting}[style=cil]
(macro m(type x)
  (type a)
  (allow x x (file (read))))
(block A
  (call m(a)))
\end{lstlisting}
seems impossible to resolve, because the type \ttype{a} defined inside \ttype{m} is passed to \ttype{m} itself as a parameter.
However, this is not deemed to be erroneous according to the compiler's behaviour.
Namely, the type \ttype{a} is copied from the macro \ttype{m} to the block \ttype{A} and then passed as parameter to \ttype{m} itself.
In a similar puzzling way, if another type named \ttype{a} is defined, e.g., in the global environment, it is shadowed by \mbox{the type copied from the macro.}

We also found cases that 
are meaningless, but are not detected as such
by the compiler.
In particular this is when typeattributes are recursively defined in a vacuous manner.
Consider for example the following configuration:
\begin{lstlisting}[style=cil]
(type a)
(typeattribute b)
(typeattribute c)
(typeattributeset b (not c))
(typeattributeset c b)
(allow b b (file (read)))
(allow c c (file (read)))
\end{lstlisting}
The typeattribute \ttype{b} should contain all the elements that are not in itself, which is a contradiction.
This error is not detected by the compiler, and a kernel policy is produced whose behaviour cannot be predicted using what we know about the semantics.
In fact, according to the compiler, \ttype{a} belongs to \ttype{b} but not to \ttype{c}, which is again contradictory since \ttype{c} is defined to be the same as \ttype{b}.
Note that such misconfigurations may arise silently in complex code where typeattributes are set using macros in different places in the code.
Indeed, we found such cases in the openWRT configuration that we used for assessing our tool.
Our tool warns the administrator about such misconfigurations and approximates the configuration behaviour by pruning the recursion tree to remove circularity.
This misbehaviour from the compiler deserves further investigation.

\paragraph{Formal semantics of CIL}

To clarify the behaviour of CIL configurations, and to formally support \lang\ and its verification mechanism, we 
provide a formal semantics for CIL.
Our semantics focuses on the type enforcement fragment of the language, which is its most used part (see the real-word CIL configurations in Section~\ref{sec:validation}),  and  maps each system type to its set of permissions.

%
     
In this section, we provide a high-level overview of our CIL semantics. Its detailed formalization is given in Appendix~\ref{app:CIL}.

Our semantics benefits from a normal form for configurations.
Roughly, we resolve inheritance and macro calls and fully qualify all names.
We compute this normal form using the following rewriting pipeline. 
This pipeline consists of six phases, where each phase repeatedly applies a set of rewrite transformations until the fixed point is reached.
\begin{enumerate}
\item The block names in \ttype{blockinherit} rules are resolved locally, if possible, or globally otherwise.
\item \ttype{blockinherit} rules are replaced by the content of the blocks they refer to.
\item The names of macros in \ttype{call} rules are resolved locally, if possible, or globally otherwise.
\item The declarations of types and typeattributes are copied from the body of the macros in the calling blocks.
\item\label{ph:calls} Macro calls are resolved: the type names in the parameters of \ttype{call} rules are resolved locally, if possible, or globally otherwise; then the \ttype{allow} rules are copied from the macros in the calling blocks. 
While copying, the non-local names in the \ttype{allow} rules are resolved in the block containing the macro definition, if possible; otherwise the resolution is delegated to further application of~(\ref{ph:calls}), until no longer possible, and then to~(\ref{ph:allow});
\item\label{ph:allow} The names in \ttype{allow} and \ttype{typeattributeset} rules in blocks are resolved locally, if possible, or globally if not.
\end{enumerate}

\noindent
The configuration in the second example is transformed by the first four phases into the left configuration below, where the \ttype{(blockinherit A)} first becomes \ttype{(blockinherit .A)} and then is resolved as \ttype{(call m(a))}; the macro name in the two occurences of \ttype{(call m(a))} are both resolved to \ttype{.m}; finally the type definition \ttype{(type b)} is copied from the macro to blocks \ttype{A} and \ttype{B}.
Phase~(\ref{ph:calls}) copies the \ttype{allow} rule instantiating the parameter \ttype{x} to the names \ttype{.a} in \ttype{A} and \ttype{.B.a} in \ttype{B}.
Finally, the two occurences of \ttype{b} are resolved to \ttype{.A.b} and \ttype{.B.b}.
Note that this representation is that of the binary representation, where names are always fully qualified.
The resulting configuration is on the right below.
\\

\begin{tabular}{ll}
\hspace{-0.8cm}
\begin{lstlisting}[style=cil]
(type a)
(macro m((type x))
  (type b)
  (allow x b 
    (file (read))))
(block A
  (type b)
  (call .m(a)))
(block B
  (type a)
  (type b)
  (call .m(a)))
  
\end{lstlisting}
&
\hspace{-0.4cm}
\begin{lstlisting}[style=cil]
(type a)
(macro m((type x))
  (type b)
  (allow x b 
    (file (read))))
(block A
  (type b)
  (allow .a .A.b (file (read)))
(block B
  (type a)
  (type b)
  (allow .B.a .B.b (file (read)))
\end{lstlisting}
\end{tabular}

\


Given a configuration in normal form, our semantic function represents it as a directed labelled graph $G = (N, \ta, A)$.
The nodes $N$ model the types and the typeattributes (with global names), and the function $\ta \colon N \rightarrow 2^N$ represents the types contained in a typeattribute (assuming $\ta(n) = \{n\}$ when $n$ is a type, which will be always the case in our examples).
The arcs $A \subseteq N \times 2^O \times N$ model permissions, where $O$ is the set of SELinux operations; we assume that whenever the typeattribute $m$ operates on $m'$, there are also the arcs $(n, o, n')$, for all $n \in \ta(m)$ and $n' \in \ta(m')$.
The meaning of $(n, o, n')$ is that the type $n$ is allowed to perform all operations in $o$ on the resources of type $n'$.
The formal definition of the semantic function is straightforward.

For example, the configuration above is associated with the following graph, where $\ta$ maps a node into the singleton set containing itself and we omit $\{\}$ for singleton sets on the arcs.

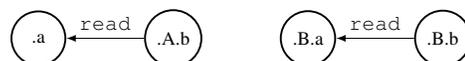
\begin{center}
                     \resizebox{0.35\textwidth}{!}{
\begin{tikzpicture}
\node[font=\footnotesize, circle, draw, thick, yshift=0cm, xshift=0cm, minimum size=8mm] (a) {.a};
\node[font=\footnotesize, circle, draw, thick, yshift=0cm, xshift=2cm, minimum size=8mm] (Ab) {.A.b};
\node[font=\footnotesize, circle, draw, thick, yshift=0cm, xshift=4cm, minimum size=8mm] (Ba) {.B.a};
\node[font=\footnotesize, circle, draw, thick, yshift=0cm, xshift=6cm, minimum size=8mm] (Bb) {.B.b};

\begin{scope}[>=latex]
\draw[<-, black] (a) to node[above, font=\footnotesize] {\ttype{read}} (Ab);
\draw[<-, black] (Ba) to node[above, font=\footnotesize] {\ttype{read}} (Bb);
\end{scope}
\end{tikzpicture}                                }
\end{center}
%

\paragraph{Adequacy of the formalization}

We define the CIL formal semantics to reflect both the implicit semantics given by the reference manual and the operational semantics defined by the compiler.
However, the documentation and the compiler sometimes disagree.
In addition, the manual has both underspecified and ambiguous cases.
When these mismatches occur and when unexplainable behavior arise, we asked CIL developers about the intended behavior~\cite{mail-log1, mail-log2}.
Some cases have been recognized as compiler bugs and the developers will fix them, whereas they will update the documentation in other cases~\cite{mail-log3}.

We identified name resolution as the most involved part of CIL's semantics, especially when 
name resolution interacts with inheritance or  macros.

As an example of a mismatch between the compiler and the reference manual, consider the following configuration:
\begin{lstlisting}[style=cil]
(block A
  (type a)
  (macro m ()
    (type a)
    (allow a a (file (read)))))
(block B
  (call A.m))
\end{lstlisting}
According to the manual, types defined inside the macro should be checked before those defined in the namespace where the macro is defined.
Hence, when copying the \ttype{allow} rule from \ttype{m} to \ttype{B}, we expect the type \ttype{a} to be resolved as \ttype{B.a}.
But it is resolved instead as \ttype{A.a}.
The CIL developers agreed that this is a bug of the compiler~\cite{mail-log3}.

The reference manual lacks a description of how the composition of CIL constructs behaves.
In particular, the composition of macro calls and block inheritance behaves differently, depending on the order in which they are resolved.
The beginning of this section presented several configurations with this kind of problem.
Since the manual specifies no evaluation order and even ignores this problem, we based the adequacy of this part of the semantics entirely on the compiler and on the developers' feedback.

To understand how to correctly compose the semantics of the different constructs, we performed comprehensive testing, discriminating between different orders.
Our tests indicate that macro calls are handled after block inheritance (i.e., phases (1) and (2) are executed before phases (3) to (6)).
We discovered that no order works for the resolution of different occurrences of block inheritance, and the same applies with different occurrences of macro calls.
This is because different occurrences of the same construct are resolved in an interleaved manner.
In other words, this resolution consists of a number of steps that are executed in the given order for all the occurrences.
For example, all the occurrences of block inheritance must complete phase (1) before any of them starts phase (2).
%
%
Note that our semantics may seem counterintuitive in some corner cases like those mentioned in paragraph~\ref{sec:ambiguity}, but it is in agreement with the developers' intent.

\section{The policy language \lang}\label{sec:IFL}
%
This section introduces IFL, our DSL for defining annotations that enable administrators to express infomation flow control requirements.
We integrate IFL with CIL, obtaining the policy language \lang, where annotations are composed with CIL constructs.
In addition, we endow \lang\ with a mechanism for statically checking that configurations satisfy their requirements.

\subsection{IFL}\label{subsec:ifl}

The constructs of IFL consider SELinux entities, typically types, and the flow of information between them.
Using IFL we define both functional requirements, allowing authorized information flows, and security requirements, preventing dangerous information flows.

\paragraph{The language}
We use IFL to model  
how information flows from one node of the graph associated with a type by the semantics, to another node, by listing the traversed nodes in the graph, and the operations allowed on them.
This is done by defining a flow \textit{kind} $P$ using the following grammar.
%
\begin{smalign*}
P ::= \ttype{n\ [o]> n}'  \mid \ttype{n\ +[o]> n}' \mid P_1\ P_2
\end{smalign*}
In this grammar, \ttype{n} and \ttype{n}$'$ are the starting and the ending nodes in a path of length $1$ for \ttype{[o]>}, and of length $1$ or longer for \ttype{+[o]>}. Nodes may also be given using the wildcard \ttype{*} standing for any node representing a type.
The non-empty set $o \subseteq O$ contains a subset of the applicable operations, and it is omitted when it is the entire set $O$.
The labeled path $P_1\ P_2$ is additionally constrained so that the ending point of $P_1$ matches the starting one of $P_2$.

The direction of arrows reflects how information flows in the graph, e.g., \ttype{n [write,read]> n$'$} means that information flows from \ttype{n} to \ttype{n$'$} when \ttype{n} writes on \ttype{n$'$} or \ttype{n$'$} reads from \ttype{n} (the operations in square brakets are the only applicable ones in this step).
A direct information flow is represented as a single step \ttype{n > n$'$}, whereas an indirect information flow is represented by multiple steps 
\mbox{\ttype{n +> n$'$}.}
A kind can also mention intermediate steps,
e.g., \ttype{n > * > n$''$ +> n$'$} specifies that information flows in two steps (through an unspecified node) from \ttype{n} to \ttype{n$''$} and then in multiple steps to \ttype{n$'$}.


Kinds are used to constrain the admissible paths of a configuration.
Given the semantics $G=(N, \ta, A)$ of a configuration, the following construction builds an \emph{information flow diagram}, i.e., a directed graph $I = (N, \ta, E)$, where the arcs of $E$ are built as follows.
For any arc $(n, o, n') \in A$, $E$ contains: $(i)$ the arc $(n, o', n')$, where $\emptyset \neq o' \subseteq o$ are the operations of $n$ on $n'$ that cause an explicit information flow from $n$ to $n'$ (e.g., \ttype{write}); $(ii)$  the arc $(n', o'', n)$, where $\emptyset \neq o'' \subseteq o$ are the operations of $n$ on $n'$ that cause an explicit information flow from $n'$ to $n$ (e.g., \ttype{read}).

The administrator can state the requirements on configurations given by the following grammar
\vspace{-1mm}
\begin{smalign*}
\mathcal R ::= P \mid \ttype{\textasciitilde}P \mid P \ttype{:} P',
\end{smalign*}
\noindent
for assertions about the information flow diagram $I$ and flow kinds.
In particular, the first type of requirement, $P$, is \emph{path existence}, which stipulates the existence in $I$ of a path $\pi$ of kind $P$.
The second, $\ttype{\textasciitilde}P$, specifies \emph{path prohibition} and requires that there are no paths in $I$ of kind $P$.
The third is \emph{path constraint} and requires that every path $\pi$ of kind $P$ in $I$ is also of kind $P'$.
 
%
%
\begin{figure}[tb]
\begin{center}
                     \resizebox{0.4\textwidth}{!}{
\begin{tikzpicture}
\node[font=\footnotesize, circle, draw, thick, yshift=0cm, xshift=0cm, minimum size=8mm] (DB) {.DB};
\node[font=\footnotesize, circle, draw, thick, yshift=0cm, xshift=3.75cm, minimum size=8mm] (http) {.http};
\node[font=\footnotesize, circle, draw, thick, yshift=1.75cm, xshift=0cm, minimum size=8mm] (anon) {.anon};
\node[font=\footnotesize, circle, draw, thick, yshift=1.75cm, xshift=4.25cm, minimum size=8mm] (home) {.home};
\node[font=\footnotesize, circle, draw, thick, yshift=1.75cm, xshift=7cm, minimum size=8mm] (other) {.other};
\node[font=\footnotesize, circle, draw, thick, yshift=0cm, xshift=7cm, minimum size=8mm] (net) {.net};

\begin{scope}[>=latex]
\draw[<-, black] (DB) to node[left, font=\footnotesize] {\ttype{read}} (anon);
\draw[->, dashed, gray, bend right] (DB) to node[right, font=\footnotesize] {\ttype{read}} (anon);

\draw[<-, black] (home) to node[left, font=\footnotesize, yshift=0.1cm] {\ttype{read}} (http);
\draw[->, dashed, gray, bend left] (home) to node[right, font=\footnotesize, yshift=0.15cm] {\ttype{read}} (http);

\draw[<-, black] (DB) to node[above, font=\footnotesize] {\ttype{write}} (http);
\draw[<-, dashed, gray, bend right] (DB) to node[above, font=\footnotesize] {\ttype{write}} (http);

\draw[<-, black] (anon) to node[above, font=\footnotesize] {\ttype{read}} (http);
\draw[->, dashed, gray, bend left] (anon) to node[above, font=\footnotesize] {\ttype{read}} (http);

\draw[<-, black] (net) to node[above, font=\footnotesize] {\ttype{read}, \ttype{write}} (http);
\draw[<-, dashed, gray, bend left] (net) to node[above, font=\footnotesize] {\ttype{write}} (http);
\draw[->, dashed, gray, bend right] (net) to node[above, font=\footnotesize] {\ttype{read}} (http);

\draw[->, black, dotted] (other) to node[right, font=\footnotesize, yshift=0.15cm] {} (home);
\end{scope}
\end{tikzpicture}
                                 }
                                 \vspace{-3mm}
\end{center}
\caption{A simple configuration (black solid arcs) and its information flow diagram (gray dashed arcs); the dotted arc represents inclusion of the target in the typeattribute of the source.}\label{fig:conf+diag}
\end{figure}
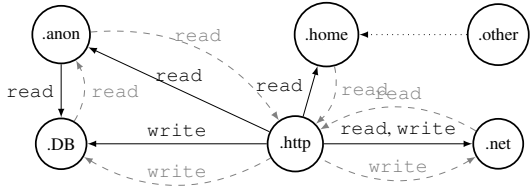
\figurename~\ref{fig:conf+diag} shows the graph semantics of a simple configuration (with the black solid arcs) and its information flow diagram (with the gray dashed arcs).
A dotted arc from a node $t$ to a node $t'$ indicates that $t'$ is in $ta(t)$.
We will further discuss this configuration in Figure~\ref{fig:annotation}.
Intuitively, the entities of type \ttype{http} collect information from the network into the database and make data available to the network and to additional entities of type \ttype{home}.
Information can flow from the network into the database and vice versa, as the configuration satisfies the functional requirements \ttype{net +> http +> DB} and \ttype{DB +> http +> net} (passing through \ttype{anon}).
Moreover, the following security requirements are met: \ttype{\textasciitilde (DB +> other)} and \ttype{DB +> net : DB > anon +> net}.
The first states that no information flows from the database to the generic, untrusted types in \ttype{other}; the second requirement says that the private information in the database passes through \ttype{anon} (where, for example, anonymization takes place) before being delivered in the network.

\paragraph{Formal semantics}
We formalize next when a configuration satisfies a given requirement.
We define a path $\pi$ of an information flow diagram $I$, and when the path $\pi$ is of kind $P$.
Intuitively, this holds when the information flow passes through the specified nodes in the correct order as a result of the designated operations.
%
\begin{definition}[information flow path and kinds]\label{def:ifpath-kind}
Let $I = (N, \ta, E)$ be an information flow diagram, a \emph{path} in $I$ is the non-empty sequence
\vspace{-1mm}
\begin{smalign*}
\pi = (n_1, o_1, n_2) (n_2, o_2, n_3) ... (n_i, o_i, n_{i+1}).
\end{smalign*} 
We say that $\pi$ has kind $P$ in $I$, in symbols $\pi \triangleright_I P$, iff 
\vspace{-1mm}
\begin{smalign*}
(\mathtt{n}, o, \mathtt{n'}) \triangleright_I  \mathtt{m}\ \ttype{[}o'\ttype{]>}\ \mathtt{m'} 
          \ \  & \text{ iff }\   (m = *\ \lor\ \mathtt{n} \in \ta(m)) \land
         \\  & \qquad (m' = *\ \lor\ \mathtt{n'}\in \ta(m')) \land
          \\ 
           & \qquad  o \cap o' \neq \emptyset 
          \\
(\mathtt{n}, o, \mathtt{n'}) \triangleright_I  \mathtt{m}\ \ttype{+[}o'\ttype{]>}\ \mathtt{m'} 
          \ \  & \text{ iff }\  
          (\mathtt{n}, o, \mathtt{n'}) \triangleright_I  \mathtt{m}\ \ttype{[}o'\ttype{]>}\ \mathtt{m'} 
 \\      
(\mathtt{n}, o, \mathtt{n'})\, \pi \triangleright_I  \mathtt{m}\ \ttype{+[}o'\ttype{]>}\ \mathtt{m'} 
          \ \  & \text{ iff }\   (\mathtt{n}, o, \mathtt{n'}) \triangleright_I  \mathtt{m}\ \ttype{[}o'\ttype{]>}\ *  \land
          \\   
          &  \qquad  \pi \triangleright_I  * \ \ttype{+[}o'\ttype{]>}\ \mathtt{m'} 
          \\          
\pi \triangleright_I P_1\ P_2 
          \  \  & \text{ iff }\   \exists \pi', \pi''.\ \pi = \pi' \pi'' 
            \land\, 
            \\ & \qquad \pi' \triangleright_I P_1  \,\land\, \pi'' \triangleright_I P_2 
\end{smalign*}
\vspace{-1mm}
\end{definition}

%
%
The second part of the definition has four cases.
The first case considers a path in the information flow diagram made of a single arc of a simple kind; since there exists an operation \ttype{op}$\,\in o \cap o'$, the arc (\ttype{n, op, n'}) can be followed transferring information from \ttype{n} to \ttype{n'}.
The second case reduces $\ttype{+[}o'\ttype{]>}$ to $\ttype{[}o'\ttype{]>}$. 
The third case simply iterates the checks along a path longer than one.
In the final case, we split a path into a prefix satisfying $P_1$ and a suffix satisfying $P_2$.
Recall that the wildcard $*$ stands for any node and can replace $n, n'$, and $n''$ above.
For example, the first clause can be rewritten as
$(\mathtt{n}, o, \mathtt{n'})\, \triangleright_I  *\ \ttype{[}o'\ttype{]>}\ \mathtt{n'} \text{ iff }  o \cap o' \neq \emptyset$
holds because the kind $*\ \ttype{[}o'\ttype{]>}\ \mathtt{n}'$ says that information flows from any node to $\mathtt{n'} $.

The predicate $I \models {\mathcal R}$ defined below expresses that a configuration with information flow diagram $I$ satisfies the requirement ${\mathcal R}$.
\begin{definition}[validity of a configuration]\label{def:requirementsem}
Let $I$ be an information flow diagram, and let $\mathcal R$ be a requirement of a given configuration.
We define $I$ to be valid w.r.t.~$\mathcal R$, in symbols $I \models {\mathcal R}$, by cases on the syntax of $\mathcal R$ as follows:
\begin{smalign*}
I \models P & \text{ iff } \exists \pi \text{ in } I \text { such that } \pi \triangleright_I P
\\
I \models \ttype{\textasciitilde} P & \text{ iff } \lnot (I \models P)
\\
I \models P_1:P_2 & \text{ iff } \forall \pi \text{ in } I \text { if } \pi \triangleright_I P_1
           \text{ then } \pi \triangleright_I P_2
\end{smalign*}
\end{definition}
It is immediate to verify that the requirements on the configuration in \figurename~\ref{fig:conf+diag} are indeed satisfied.

\paragraph{Expressivity}

A path existence constraint expresses a functional requirement, namely that a specific information flow is allowed.
If satisfied, this constraint ensures the administrator that the configuration does not prevent the system from performing the desired
task.
In contrast, a prohibition constraint specifies a security requirement: a configuration obeying it never goes wrong.
For example, one can easily specify confidentiality in a Bell-La Padula style, or integrity in the Biba integrity model.
Finally, path constraints can express nontransitive properties, like intransitive noninterference.
For example $\mathtt{n \ttype{+>} n' : n \ttype{+>} n'' \ttype{+>} n'}$ requires that the type 
$\mathtt{n}$ cannot transmit any information to $\mathtt{n'}$ unless it is done through $\mathtt{n''}$.

IFL can express (positive and negative) reachability properties with constraints on the paths. 
Since information flow diagrams are labeled transitions systems, IFL has similarities to temporal logics.
Actually, IFL kinds can be expressed as LTL formulas, as the encoding in Section~\ref{sec:LTLencoding} shows.
However, IFL allows an additional quantification over paths, which can appear only at top level and is not \mbox{expressible in LTL.}

\subsection{\lang}

We introduce the language \lang, obtained by integrating IFL into CIL.
More precisely, we add the following two constructs to augment CIL with comments that specify IFL requirements.
\begin{enumerate}
\item 
\emph{Information flow requirement definitions}, which may occur in blocks and macros.
We use them to specify IFL requirements with labels, which must be satisfied by the allow rules of the configuration where they occur.
Requirements are copied when calling a macro or inheriting a block, and are managed coherently with the other rules, e.g., concerning name resolution.
\item
\emph{Refinement of requirements}, which may occur within \ttype{call} and \ttype{blockinherit} instructions.
Refinements strengthen requirements by further elaborating  constraints in the inheriting or caller block.
\end{enumerate}
To ensure backward compatibility, requirement definitions and refinements are enclosed between \ttype{;IFL;} thereby taking the form of CIL comments.

The example in \figurename~\ref{fig:annotation} illustrates both constructs.
For example, the second line in \figurename~\ref{fig:annotation} contains a functional requirement labeled with \ttype{(F1)} that requires the existence of a direct or indirect information flow from the node \ttype{inp} to \ttype{out}.
Nodes in the IFL requirements are types, and are resolved as any other CIL name, e.g., the parameter \ttype{out} is bound to \ttype{net} in the first call of the macro \ttype{in\_out}; in contrast, the requirements \ttype{F1} and \ttype{F2} are simply copied.

In the second call, the administrator duplicates the requirements and refines them with further constraints about how information must flow, since the intermediate node \ttype{http} is inserted in the requirements.
Note that the new labels refer to those of the original requirements.
The new requirements impose that a flow must exist from \ttype{net} (instantiating the parameter \ttype{inp}) to \ttype{DB} (instantiating \ttype{out}) and vice-versa, both passing through \ttype{http}.
Note that since the wildcard is used there is no constraint on the actual parameters.
The refined requirement \ttype{(F1R:F1)} results then in the path constraint 
\ttype{net +> http : http +> DB}, while the refined requirement \ttype{(F2R:F2)} is \ttype{DB +> http : http +> net}.

Similarly, in the call to \ttype{anonymize}, the requirement \ttype{(S1)} is refined by specifying that the operation in the single step is \ttype{read}.
Of course, the same happens when inheriting a block.
Finally, the requirement \ttype{(S2)} states that information cannot flow from \ttype{DB} to \ttype{other}.
\begin{figure}
\begin{lstlisting}[style=cil]
(macro in_out((type inp) (type out))
  ;IFL; (F1) inp +> out ;IFL;
  ;IFL; (F2) out +> inp ;IFL;)
(macro anonymize((type x) (type y))
  (type anon)
  (allow anon x (file (read)))
  ;IFL; (S1) x +> y : x > anon +> y ;IFL;)

(typeattribute other)
(typeattributeset other 
  (not (or DB (or http (or anon net)))))

(type DB)
(type http)
(type home)
(type net)

(call in_out(net http))

(call in_out(net DB)
  ;IFL; (F1R:F1) * +> http +> * ;IFL;
  ;IFL; (F2R:F2) * +> http +> * ;IFL;)
	
(call anonymize(DB net)
  ;IFL;(S1R:S1) DB+>net : DB[read]>anon+>net ;IFL;)

(allow http anon (file (read)))
(allow http DB (file (write)))
(allow http other (file (read)))
(allow http net (file (read write)))
;IFL; (S2) ~ DB +> other ;IFL;
\end{lstlisting}
\vspace{-3mm}

\caption{Example of CIL configuration with IFL annotations.}
\label{fig:annotation}
\end{figure}

We now introduce the most important details of the formalization of \lang;
its complete definition is in the Appendix~\ref{app:IFCIL}.
%
We first discuss the notion of refinement of IFL requirements.
Intuitively, a refinement of a requirement $\mathcal{R}$ allows a subset of the information flow paths allowed by $\mathcal{R}$.
This is formalized by a preorder $\preceq$, saying that $\mathcal{R}' \preceq \mathcal{R}$ if $\mathcal{R}'$ refines $\mathcal{R}$; the precise definition of $\preceq$ is given in Appendix~\ref{app:IFCIL}.

We prove the following theorem, stating that the validity of configurations is preserved by refinement.
\begin{restatable}[Refinement]{theorem}{refinement}\label{thm:ref}
Let $I$ be an information flow diagram, and let $\mathcal{R}'$ and $\mathcal{R}$ be two IFL requirements such that $\mathcal{R}' \preceq \mathcal{R}$.
Then
\vspace{-3mm}
\begin{smalign*}
I \models \mathcal{R}' \Rightarrow I \models \mathcal{R}.
\end{smalign*}
\end{restatable}
\vspace{-3mm}

In defining the semantics of \lang, we use the meet of two requirements $\mathcal{R}' \sqcap \mathcal{R}$ on the set of requirements preordered with $\preceq$, i.e., the largest requirement w.r.t.\ $\preceq$ that is smaller than both $\mathcal{R}'$ and $\mathcal{R}$.
To see why, consider the following requirements taken from the example above.
\begin{lstlisting}[style=cil]
  ;IFL; (F1) inp +> out ;IFL;
  ;IFL; (F1R:F1) * +> http +> * ;IFL;
\end{lstlisting}
These requirements are incomparable with respect to $\preceq$.
To see this, take $I$ and $I'$ with nodes in $\{$\ttype{inp, out, http, a}$\}$ such that $I$ has a single arc (\ttype{inp, out}), and $I'$ only has the two arcs (\ttype{a, http}) and (\ttype{http, a}).
If they were comparable, Theorem~\ref{thm:ref} would be falsified because \mbox{$I \models$ \ttype{inp +> out}} but $I \not\models$ \ttype{* +> http +> *}, and similarly for $I'$ replacing $I$.
Although incomparable, \ttype{F1} and \ttype{F1R:F1} are clearly related.
Namely there exists the meet of the two $\ttype{F1} \sqcap \ttype{F1R:F1} = \ttype{inp +> http +> out}$.
This meet has more details than, and refines both, \ttype{F1} and \ttype{F1R:F1}, 
because it requests the presence of an information flow from the node \ttype{inp} to \ttype{out}, via \ttype{http}.

We are now ready to define the semantics of \lang.
We first normalize configurations by applying the six transformation phases described in Section~\ref{sec:formal}, taking meets whenever needed.

The semantics of a configuration consists of a graph $G$ and a set of requirements $\mathbb{R}$ representing the semantics of a CIL configuration and the IFL annotations.
It is defined as:
\begin{definition}[\lang\ semantics]
Given a (normalized) \lang\ configuration $\Sigma$, its semantics is the pair $(G, \mathbb{R})$, where $G$ is the CIL semantics of $\Sigma$ and $\mathbb{R}$ is the set of IFL requirements occurring in $\Sigma$.
\end{definition}

Not all configurations satisfy their requirements, and we define below when they do, i.e., when the information flow respects the constraints expressed in the IFL annotations. 
\begin{definition}[correct \lang\ configuration]
Let $\Sigma$ be a (normalized) \lang\ configuration, let $(G, \mathbb{R})$ be its semantics, and let $I$ be the information flow diagram of $G$. The configuration $\Sigma$ is correct, in symbols $I \models \mathbb{R}$, 
iff $I \models \mathcal{R}$ for all $\mathcal{R} \in \mathbb{R}$.
\end{definition}

\section{Requirement verification}\label{sec:verification}
%

We describe next how we automatically check that a \lang\ configuration respects the given information flow requirements.
We rely on model checking, so as to reuse existing verification tools.
For this, we first encode a configuration as a Kripke transition system~\cite{Muller-OlmSS99} and 
an IFL requirement 
as an LTL formula.

\subsection{Encoding in temporal logic}\label{sec:LTLencoding}

A Kripke transition system (KTS) over a set $AP$ of atomic propositions is $K = (S, Act, \rightarrow, L)$, where $S$ is a set of states, $Act$ is a set of actions, $\rightarrow\, \subseteq S \times Act \times S$ is a transition relation, and $L \colon S \rightarrow 2^{AP}$ is a labeling function mapping nodes to a set of proposition that hold at that node.
Paths of $K$ are defined as alternating sequences of states and actions starting and ending with a state.

We associate an IFCL configuration $\Sigma$ with a KTS with the nodes of $\Sigma$ as states and
the edges of the information flow diagram of $\Sigma$ as transitions (for technical reasons, transitions are labeled with a single operation), and the type and typeattribute names of $\Sigma$ as atomic propositions.

%
\begin{definition}[Encoding of configurations]
	Let $I= (N, \ta, E)$ be the information flow diagram of a configuration $\Sigma$. 
	The corresponding KTS is $K = (N, O, E', \Lambda)$, where
	%
	\begin{itemize}
		\item $O$ is the set of SELinux operations
		\item $E' = \{ (n, op, n') \mid (n, o, n') \in E \land op \in o \}$
		\item \mbox{$M\!\in\!\Lambda (n)$  if $n\!\in\!\ta(M)$, i.e., $n$ is in the typeattribute $M$}
	\end{itemize}
\end{definition}


We encode IFL kinds in a suitable version of LTL~\cite{Muller-OlmSS99}, where the syntax of formulas $\phi$ is
\begin{smalign*}
\phi ::= p \mid (op) \mid \phi_1 \land \phi_2 \mid \phi_1 \lor \phi_2 \mid \lnot \phi \mid X(\phi) \mid \phi_1 U \phi_2.
\end{smalign*}
We write $w \models_l \phi$ if the path $w$ of $K$ satifies the LTL formula $\phi$; the formal definition is standard and can be found in~\cite{Muller-OlmSS99}.
Intuitively, $w$ satisfies the atomic proposition $p$ if it starts with a node labeled with $p$; 
$w$ satisfies $(op)$ 
if its first action is $op \in O$; conjunction, disjuction, and negation are as usual; $X(\phi)$ is 
satisfied by
$w$ if its subpath starting from the second state satisfies $\phi$; 
and $w$ satisfies $\phi_1 U \phi_2$
if there exists a node $s$ in $w$ such that 
the subpath starting from it satisfies $\phi_2$
and every subpath starting from a state before $s$ 
satisfies $\phi_2$.

For convenience, in the following we simplify our grammar for the flow kind $P_1 P_2$ and rewrite the grammar from subsection~\ref{subsec:ifl} in the following equivalent form (recall that the starting node of $P$ in the last two cases is \ttype{n}$'$).
\begin{smalign*}
\!\!\!P ::= \ttype{n\,[o]> n'}  \mid \ttype{n\,+[o]> n}' 
                \mid 
                (\ttype{n\,[o]> n'}) P  \mid (\ttype{n\,+[o]> n}') P
\end{smalign*}

\begin{definition}[Encoding of flow kinds]\label{def:enc-kind}
The encoding of flow kinds is defined as follows.
\begin{smalign*}
&\ltlencode{\ttype{n [o]> n}'} = \ttype{n} \land \bigvee_{\ttype{op} \in \ttype{o}} (\ttype{op}) \land X(\ttype{n}' \land \lnot X(\mathit{true}) )\\
&\ltlencode{\ttype{n +[o]> n}'} = \ttype{n} \land \bigvee_{\ttype{op} \in \ttype{o}} (\ttype{op}) \land X(\bigvee_{\ttype{op} \in \ttype{o}} (\ttype{op})\ U\ (\ttype{n}' \land \lnot X(\mathit{true})))\\
&\ltlencode{(\ttype{n [o]> n}') P} = \ttype{n} \land \bigvee_{\ttype{op} \in \ttype{o}} (\ttype{op}) \land X(\ltlencode{P})\\
&\ltlencode{(\ttype{n +[o]> n}') P} = \ttype{n} \land \bigvee_{\ttype{op} \in \ttype{o}} (\ttype{op}) \land X(
\bigvee_{\ttype{op} \in \ttype{o}} (\ttype{op})\ U\ \ltlencode{P})
\end{smalign*}
\end{definition}
Note that we use $\lnot X(true)$, i.e., the path is complete as there is no next step, to represent the fact that IFL semantics is defined on finite paths.

LTL's semantics can be lifted to a KTS $K$ in quite different manners, with the best for modeling IFL being as follows:
$K \models_{l} \phi$ iff $\forall w \in W . w \models_{l} \phi$, where $W$ is the set of all (finite and infinite) paths in $K$.
This paves the way for defining when a KTS satisfies a set of IFL requirements.
%
%
%
\begin{definition}[Satisfaction of configurations]\label{def:requirementsat}
Let $\mathcal{R}$ be a requirement of a given configuration, let $K$ be a KTS, and let $W$ be the set of paths in $K$.
We define the satisfaction relation $\vdash$ on the syntax of $\mathcal{R}$ as follows:
\begin{smalign*}
K \vdash P &\text{ iff } K \not\models_l \lnot\ltlencode{P}\\
K \vdash \ttype{\textasciitilde}P &\text{ iff } K \models_l \lnot\ltlencode{P}\\
K \vdash P \ttype{:} P' &\text{ iff } K \models_l \lnot\ltlencode{P} \lor \ltlencode{P'}
\end{smalign*}
%
%
\noindent
We homomorphically extend $\vdash$ to sets of requirements.
\end{definition}
Note that the three clauses above mimic the analogous clauses in the Definition~\ref{def:requirementsem}.
The first clause says that at least one path satisfies $\ltlencode{P}$; conversely, the second clause says that no path satisfies $\ltlencode{P}$.
The third clause simply contains the boolean definition of classical implication.

This correspondence supports the correctness of our verification technique, which is expressed by the following theorem stating that the notions of validity and  satisfaction of configurations coincide:
\begin{restatable}[Correctness and Completeness]{theorem}{correctness}\label{thm:correctness}
Let $\Sigma$ be an IFCIL configuration with requirements $\mathbb{R}$, let $I$ be its information flow diagram, and let $K$ be the KTS of $\Sigma$. Then
\vspace{-1mm}
\begin{smalign*}
K \vdash \mathbb{R}\ \text{  if and only if  }\ I \models \mathbb{R}.
\end{smalign*}
\end{restatable}
\vspace{-2mm}
\subsection{Model checking IFCIL}

Theorem~\ref{thm:correctness} enables us to reuse model checking techniques and tools to automatically verify that a configuration is correct with respect to its information flow requirements.
In particular, an LTL model checker provides us with a decision algorithm for $K \models_{l} \phi$.

In this work, we resort to the classical model checker \modelchecker\ that targets only infinite paths, as usual, whereas Definition~\ref{def:requirementsat} also considers finite paths.
Therefore, we extend $K$ to a KTS $K_\iota$ with a distinguished sink state $\iota$ and with additional transitions (labeled with every $op \in O$) from every state to $\iota$.
Roughly, the satisfaction relation is updated by substituting $X(\iota)$ for $\lnot X(true)$.
Theorem~\ref{thm:correctness} still holds (see Corollary~\ref{thm:MC} in Appendix~\ref{app:IFCIL}).
This extension enables us to verify the correctness of IFCIL configurations with \modelchecker.

The worst case complexity of LTL model checking is unfortunately $2^{\mathcal{O}(|\phi|)}\mathcal{O}(|S| + |E'|)$~\cite{Muller-OlmSS99}, where $|S|$ and $|E'|$ are the number of nodes and number of arcs in the KTS, respectively.
In practice we expect the size of the configuration and the number of requirements to grow as the system grows.
In contrast, we do not expect the size of each IFL requirement to depend on the system size.

We now specialize the formula above to our encoding.
It is easy to see that $|S|$ is equal to the number of type declarations in the configuration, and that 
$|E'|$ is bounded from above by $|S| \times |S| \times |O|$, where $O$ is the set of SELinux operations.
Note that the size of an LTL formula resulting from the encoding of an IFL requirement is linear with respect to the number of names in the requirement.
Thus, the complexity of verifying an \lang\ configuration is $\sum_{\phi \in \mathbb{R}} 2^{\mathcal{O}(|\phi|)}\mathcal{O}(|S|)^2 \times \mathcal{O}(|O|)$.
Since $|\phi|$ is usually small and does not increase with the configuration size, the complexity grows linearly with respect to the number of requirements and  operations, and quadratically with respect to the number of types.

Experiments with our prototype implementation on real-world configurations show that results are obtained in an acceptable amount of time, on the order of seconds,
see below.


\section{The tool \tool}\label{sec:tool}
%

We now describe our tool \tool\ that given a \lang\ configuration verifies its correctness with respect to its information flow requirements.
Although our tool is currently a prototype, and not optimized, we were nevertheless able to successfully apply it to large, complex, real-world policies.

\subsection{Translation to \modelchecker}
Our tool has a front end that reads a configuration, normalizes it, and then computes its semantics, the associated KTS, and the LTL representation of the requirements, expressed in the \modelchecker\ input language.
The result is supplied to the model checker \modelchecker, which checks each requirement.
Finally the administrator is notified which requirements are satisfied and which are not.

In more detail, \tool\ takes as input an \lang\ configuration and an associated file where every operation comes with the direction of the information flow it causes.
This file is used to build the information flow diagram.

The tool explicitly handles CIL's constructs for defining classes and permissions, and reduces the input configuration to one that only uses the fragment of CIL presented in Section~\ref{sec:formal}.
Since other constructs, like those concerning roles, do not affect requirement satisfiability, the tool just ignores them.


For example, the following
\begin{lstlisting}[style=cil]
(type DB)
(type http)
(type home)
(type net)
(type anon)

(typeattribute other)
(typeattributeset .other 
  (not (or .DB (or .http (or .anon .net)))))

(allow .anon .DB (file (read)))
(allow .http .anon (file (read)))
(allow .http .DB (file (write)))
(allow .http .other (file (read)))
(allow .http .net (file (read write)))

;IFL; (F1) .net +> .http ;IFL;
;IFL; (F2) .http +> .net ;IFL;
;IFL; (F1R) .net +> .http +> .DB ;IFL;
;IFL; (F2R) .DB +> .http +> .net ;IFL;
;IFL; (S1R) .DB+>.net: .DB[read]>.anon+>.net ;IFL;)
;IFL; (S2) ~ .DB +> .other ;IFL;
\end{lstlisting}
is the normalization of the configuration in \figurename~\ref{fig:annotation}.
Its \lang\ semantics is the pair $(G, \mathbb{R} = \{ \ttype{F1, F2, F1R, F2R,}$ $\ttype{S1R, S2} \})$, where $G$ is the CIL semantics in Figure~\ref{fig:conf+diag}.
It is trivial to derive the KTS $K$ associated with $G$.
To verify the satisfaction of the requirements, we check $K \vdash \mathcal{R} \in \mathbb{R}$.
We only show the case $\mathcal{R} = \ttype{S1R}$, i.e., $K \models_l \lnot \ltlencode{\ttype{.DB +> .net}} \lor \ltlencode{\ttype{.DB [read]> .anon +> .net}}$
where:
\begin{smalign*}
&\ltlencode{\ttype{.DB +> .net}} =
\ttype{.DB} \land \bigvee_{\ttype{op} \in \ttype{O}} (\ttype{op})\ \land\\
&\quad X(\bigvee_{\ttype{op} \in \ttype{O}} (\ttype{op})\ U\ (\ttype{.net} \land \lnot X(true) ))\\
&\ltlencode{\ttype{.DB [read]> .anon +> .net}} = \ttype{.DB} \land (\ttype{read})\ \land\\
&\quad X(\ttype{.anon} \land \bigvee_{\ttype{op} \in \ttype{O}} (\ttype{op}) \land X(\bigvee_{\ttype{op} \in \ttype{O}} (\ttype{op})\ U\ (\ttype{.net} \land \lnot X(true) )))
\end{smalign*}
The resulting input file for \modelchecker\ represents the nodes of the KTS by variable assignments and transitions as updates of such assignments (using the \ttype{next} operator).
\begin{lstlisting}[style=NuSMV]
MODULE main

DEFINE
  other := (!((type=DB | (type=http | 
      (type=anon | type=net))))) & !(type=sink);
VAR
  type : { sink, DB, anon, home, http, net };

IVAR
  operation : { read, write };

TRANS
  (type=DB -> 
  	((operation=read & next(type=anon)) | 
  	  next(type=sink))) &
  (type=anon -> 
  	((operation=read & next(type=http)) | 
  	  next(type=sink))) &
  (type=home -> 
  	(next(type=sink))) &
  (type=http -> 
  	((operation=write & next(type=DB)) | 
  	  (operation=write & next(type=net)) | 
  	    next(type=sink))) &
  (type=net -> 
  	((operation=read & next(type=http)) | 
  	  next(type=sink))) &
  (type=sink -> next(type=sink))

LTLSPEC (!(type=DB & X(F type=net)) | (type=DB & 
 operation=read & X(type=anon & X(F type=net))))
LTLSPEC !(type=net & X(F type=http))
LTLSPEC !(type=http & X(F type=net))
LTLSPEC !(type=DB & X(F(type=http & X(F type=net))))
LTLSPEC !(type=net & X(F(type=http & X(F type=DB))))
LTLSPEC !(type=DB & X(F other))
\end{lstlisting}
%
We briefly comment on the encoding to generate the input file:
\begin{itemize}
\item The state variable \ttype{type} has the enumeration type that lists all the types in the configuration, 
plus \ttype{sink} (i.e., $\iota$).
\item Typeattributes are encoded as symbols and defined as predicates on types.
\item The input variable \ttype{operation} has the enumeration type that lists all the operations in $O$.
\item The transitions are defined in \ttype{TRANS}: from each starting node there is an arc to the possible types and typeattributes with the appropriate operation.
\item Requirements are expressed in the syntax of \modelchecker\ as defined by $\ltlencode{\_}$.
\end{itemize}
\tool\ then parses the response of \modelchecker\ and answers positively: all the requirements are verified within few seconds.

\subsection{Validation}
\label{sec:validation}

We experimentally assessed our tool on three real-world CIL policies.
The first policy~\cite{openWRTconfig} is used in the OpenWrt project, a version of the Linux operating system targeting embedded devices, like network  appliances~\cite{openWRTproject}.
The second and the third are SELinux example policies, namely \emph{cilbase}~\cite{cilbase} and \emph{dspp5}~\cite{dspp5}, which serve as templates for creating personalized configurations.
The analyzed policies have more than ten thousands lines of code, and make extensive use of all CIL's advanced features, in particular macros and blocks.

To illustrate IFL's expressivity, we use it to formalize various properties that are often considered in the literature, as well as domain-specific policies that we designed.
Expressing them in IFL is easy and the resulting requirements are short, direct, and natural.
%
%
Moreover, we assess the scalability of \tool\ on real-world examples and show that it scales well to large configurations, checking their requirements in a few seconds.

\paragraph{Properties}
We first consider the following property inspired by Jeager et al.~\cite{Jaeger}, who investigated the \emph{trusted computing base} (TCB) of an SELinux configuration and checked from which types information flows to the TCB, identifying those that do not compromise security.
Using IFCIL, the administrator can restrict the information flows to the TCB to the permitted ones by defining the typeattributes \ttype{TCB} and \ttype{Harmless}, and by requiring \ttype{ +> TCB : Harmless +> TCB}.

The second property states that the flow from \ttype{a} to \ttype{z} must pass through a list of intermediate entities \ttype{b,c} ... ~\cite{Guttman}, also called \emph{assured pipeline}~\cite{patterns}. Here, it suffices to define requirements of the form \ttype{a +> z : * +> b +> c +> ... +> *}.

We express the \emph{wrapping of untrustworthy programs} of~\cite{patterns}, by
defining requirements stating that all the information flows from (or to) a given type \ttype{untrustworthy} must pass through a \ttype{verifier} type as first step, i.e., \ttype{untrustworthy > * : * > verifier}.

Finally, we propose the additional \emph{augment-only} property that only allows elements of type \ttype{a} to increase (\ttype{append}) the information on the targets  with type \ttype{b} without overwriting or removing any.
This property is expressed as \ttype{a > b : a [append]> b} and \ttype{a +> > b : a +> [append]> b}.

\paragraph{Experimental results}

The results of our analyses on the three configurations are summarized in Table~\ref{tab:performance}.
For each row, the table reports the kind of property, the number of requirements, and the total time for verifying them (\modelchecker\ input file generation plus LTL model checking).
The tool took approximately two minutes to check the entire OpenWRT configuration, and less then three seconds for the other two policies.
The analysis reports that some requirements are violated.
Among these, the checks on the TCB property show that information flows exist from types that are likely untrusted to types related to the OS security mechanisms, e.g., in \textit{dspp5}, information can flow from \ttype{.lostfound.file} to \ttype{.sys.fs}.
We are investigating whether these types are indeed untrusted and the actual impact the detected violations have on security.
This however requires reverse engineering to better understand the security goals of the analyzed policies.



\begin{table}
\caption{Performance analysis on tree real-world configurations}
\label{tab:performance}
\center

\begin{tabular}{|l|c|c|c|}
\hline
Property & Requirements & Verification Time \\
\hline
\multicolumn{3}{l}{}\\[-0.15cm]
\multicolumn{3}{l}{\textbf{openWRT} (45702 lines, 590 types)}\\
\hline
TCB & 1 & 119sec \\
\hline
assured pipeline & 3 & 122sec \\
\hline
wrap untrustworthy & 10 & 100sec \\
\hline
augment only & 2 & 115sec \\
\hline
total & 16 & 129sec \\
\hline
\multicolumn{3}{l}{}\\[-0.15cm]
\multicolumn{3}{l}{\textbf{cilbase} (11989 lines, 293 types)}\\
\hline
TCB & 1 & 0.240sec \\
\hline
assured pipeline & 4 & 0.238sec \\
\hline
wrap untrustworthy & 6 & 0.235sec \\
\hline
augment only & 2 & 0.232sec \\
\hline
total & 13 & 0.258sec \\
\hline
\multicolumn{3}{l}{}\\[-0.15cm]
\multicolumn{3}{l}{\textbf{dspp5} (14782 lines, 149 types)}\\
\hline
TCB & 1 & 2.22sec \\
\hline
assured pipeline & 4 & 2.14sec \\
\hline
wrap untrustworthy & 8 & 2.28sec \\
\hline
total & 13 & 2.24sec \\
\hline
\end{tabular}
\end{table}

\section{Related Work}\label{sec:related}
%

Checking the security of the information flow is a problem that has been widely studied since the seminal work by Denning~\cite{Denning}, and Goguen and Meseguer~\cite{GoguenMeseguer}. 
Below, we discuss proposals that target information flow in SELinux policies. 
We also consider research that addresses verifying information flow in access control languages, and that augment existing programming languages with information flow policies.
For a broad survey on the topic, we refer the reader to~\cite{survey}.

\paragraph{Information flow in SELinux}
Numerous tools for SELinux policy analysis have been proposed. 
Many of them are based on information flow, but none targets CIL or explicitly handles the advanced features we consider. 
These tools can be divided in two categories.
The first focuses on predefined tests, searching for specific kinds of misconfigurations.
The second supports administrators in querying information flow properties of given policies.
%
%
Since our tool enables administrators to perform custom analysis, it differs from the proposals in the first category that we briefly survey.

Reshetova et. al.~\cite{Reshetova} propose SELint, a tool for detecting common kinds of misconfigurations in given SELinux configurations, e.g., the overuse of default types, and the association of specific untrusted types with critical permissions.
In contrast to our work, their approach is also specialized for mobile devices.    
%

Radika et. al.~\cite{Radhika} analyse SELinux configurations to spot potentially dangerous information flows.
They consider an information flow from an entity $a$ to an entity $b$ to be potentially dangerous if a \ttype{neverallow} rule prohibits a direct read access from $b$ to $a$.
They propose two tools: the first statically investigates such information flows in configurations, and has been applied to the SELinux reference policy and to the Android policy~\cite{AndroidOSP};
the second is a run-time monitor that dynamically tracks information flows in an SELinux system.
Our tool does the same kinds of analysis, and also expresses more specific requirements. 
We can, for example, check for the presence of direct information flows caused by operations different from those in \ttype{neverallow}
and of intransitive information flows that pass through a specific path.

Jaeger et. al.~\cite{Jaeger} analyze the SELinux example policy for Linux 2.4.19, focusing on integrity properties.
They determine which entities are in the TCB and analyze their integrity by focusing on transitive information flow.
As discussed above, we let the administrator specify the TCB and the desired requirements while developing the configuration, rather than deriving the TCB after the policy is implemented.

We now briefly discuss the proposals in the second category that are closest to ours.
These proposals neither directly work on structured CIL configurations nor 
they offer real support for advanced features of this language. 
Moreover, they do not allow labeling configurations with requirements that interact with the language constructs.
All the properties they consider are global.
In contrast, our proposal works directly on structured CIL configurations and our requirements are first class citizens in \lang.

Guttman et al.~\cite{Guttman} propose a formal model of SELinux access control, based on transition systems, and provide an LTL model checking procedure to verify that a configuration satisfies the security goals specified by the administrator.
The security goals they consider are non-transitive information flow properties: they verify that every information flow between two given SELinux entities (e.g., users, types, roles) passes through a third entity.
As discussed above, \lang\ expresses these requirements, also with conditions about the operations occurring in the information flow.
In contrast, we do not consider exceptions as they can be encoded using typeattributes.

Sarna-Starosta et al.~\cite{sarnastarosta} propose a logic-programming based approach to analyzing SELinux policies.
Their tool transforms a configuration into a Prolog
program, thus allowing the administrator to perform deductions on the properties of the configuration with the standard Prolog query mechanism.
This proposal is similar to ours except that we target CIL and allow labels inside configurations. 
Also, they rely on libraries of predefined queries for assisting users not familiar with logic programming.
Our DSL precisely targets information flows, and easily compiles into LTL.

Finally, high-level languages have been proposed for SELinux based on information flows.
All these languages were presented prior to the introduction of CIL; they therefore target the kernel policy language and do not exploit CIL's advanced features.
In contrast, we consider an already adopted language, namely CIL, and extend it with useful features, that support administrators in reasoning about their code.
Moreover, \lang\ is backward compatible. Administrators thus neither need to change the workflow nor the tools they use to develop and maintain SELinux configurations.  

Hurd et al.~\cite{Lobster} propose Lobster, a high-level DSL for specifying SELinux configurations.
This compositional language describes the configuration's expected information flow.
Instead of macros and blocks, Lobster provides the user with class definition and instantiation, where operations and permissions are represented as ports and labeled arrows between ports, respectively.
The user must specify all the desired information flows of the system and the compiler checks that no others are possible.
In contrast, we allow the user to succinctly specify wanted and unwanted information flows.
In particular, user can also specify ``negative'' requirements that explicitly forbid some information flows, while Lobster allows specifying only the ``positive'' flows.
Moreover, \lang\ supports more fine-grained requirements, letting users choose the level of details in defining the information flow in the system, e.g., targeting only critical permissions. 
Finally, Lobster is not backward compatible with SELinux, whereas \lang\ is backward compatible. 

Nakamura et al.~\cite{Nakamura} propose SEEdit, a security policy configuration system that supports creating SELinux configurations using a high-level language called the Simplified Policy Description Language (SPDL).
SPDL keeps the configuration small because the administrator can group SELinux permissions and refer to system resources directly using their name instead of types.
They implement a converter that produces SELinux configurations, and they propose a set of tools for automatically deriving (parts of) a configuration using system logs. 
Their main objective is mainly to simplify the usage of the kernel policy language, working on its syntax and adding utility features.
\mbox{Static checking is not supported.}

\paragraph{Information flow on access control}
Bugliesi et al.~\cite{Gran} develop a verification framework supported by a tool for \texttt{grsecurity}, a role-based access control system for Unix/Linux~\cite{grsecurity}. 
They propose an operational semantics for \texttt{grsecurity} and an abstraction mechanism that reduces the problem of policy verification to reachability, thereby allowing for model checking.
The properties they address concern establishing whether a given subject can access a given resource and the writing and reading flows on resources. 
Although they address properties similar to ours, these properties are built into the verification framework and there is no language for formalizing new security requirements. 
Moreover, they do not address non-transitive properties as we do.

Calzavara et al.~\cite{Calzavara:types} continue this line of research by proposing a security type system for verifying information flow in ARBAC policies. 
In particular, their type system can address the role reachability problem and offers a compositional technique.  
Also in these policies are built into the type system and do not cover intransitive properties.
In contrast, our policies are written in a language, can be composed, and can express intransitive properties. 
Since we use model checking, the results of our verification phase cannot, however, be composed. 

Guttman and Herzog~\cite{Guttman:network} consider a network access control scenario. 
They propose a formalism for expressing networks and security goals about the trajectories a network packet can follow.
Moreover, they propose ad hoc algorithms that determine if the security goals are satisfied by a system.
Their security goals are similar to our information flow requirements in terms of the expressible properties.
However, we reduce verification to standard model checking. 

\paragraph{Adding information flow to real languages}

Numerous papers address the control of information flow in the language-based approach where programmers specify how data may be used.
Below we consider some proposals, focussing on full-fledged
security-typed languages.

FlowCAML~\cite{FlowCaml} is a variant of ML with information-flow types and type inference and provides support for the static enforcement of Denning-style confidentiality policies.

Jif~\cite{Jif} extends Java with the decentralized label model where data values are labeled
with security policies. 
The Jif compiler enforces these security policies performing some static checking. 
Moreover, Jif supports declassification, which provides a liberal information flow escape
hatch for programs that would otherwise be rejected by the compiler.

Fabric~\cite{Fabric} extends Jif with support for distributed programming and transactions. 
It provides several mechanisms for controlling accesses and information flow, to prevent  violating confidentiality and integrity policies. 
All values in Fabric are labeled with policies in the decentralized label model that express security requirements in terms of principals.
These labels allow principals to control to what extent other principals can learn or affect their information.

Lifty~\cite{liquid:type} is a domain-specific language for data-centric applications that allows programmers to annotate the sources of sensitive data with declarative information flow policies.
Lifty uses liquid types to enforce static information flow control and to statically and automatically verify that the application obeys the policies.
Moreover, its compiler is equipped with a repair engine that automatically patches any found leaks.

Paragon~\cite{paragon} extends Java with information flow policies building on an object-oriented generalisation of Paralocks~\cite{paralocks}. 
A policy is a set of flow locks that are conditions constraining how principals handle data and that can be opened and closed via instructions.
Paragon expresses a wide variety of policy paradigms, including Denning-style policies, the Jif decentralised label model, and stateful information flow policies.

All the above papers address Turing complete languages.
They encode policies through security labels and enforce them through security type systems.
We instead propose a declarative language for expressing policies that describe the admitted and prohibited flows rather than associating labels to resources and user.
Also, we target a configuration language that is not Turing complete, and our verification mechanism is based on model checking. 
Finally, in contrast to some of the above proposals, we do not explicitly deal with declassification.

\section{Conclusions and future work}\label{sec:conclude}


%
%
%


We have proposed IFL, a language for expressing fine-grained information flow requirements.
Its declarative nature makes it easy to embed it in various access control languages and facilitates requirement verification through standard model checkers.
We exploit IFL to obtain \lang, a backwards compatible extension of CIL.
\lang\ helps administrators in writing information flow policies, including confidentiality and integrity requirements.
We have also defined and implemented a verification procedure to check if an \lang\ configuration complies with its IFL requirements.
Our experiments show that the language works well for defining properties that are commonly investigated for SELinux policies, and that the verification times are acceptable even for large real-world configurations.

\paragraph{Discussion}

We believe that our extension can help with the development of more advanced high-level languages. 
As our annotations are associated with a common intermediate language, they can enrich different high-level languages.
Our verification procedure can be used for checking properties when composing code written in different languages.

Our semantics focuses on CIL type enforcement because it allows defining more fine-grained information flow policies than other constructs, like those for multi-level security~\cite{SELinuxMLS}.
Moreover, many real CIL configurations only use these more limited constructs.
We do not explicitly model the constructs for defining the operations used inside allow rules.
But this is not a limitation because these constructs can be easily encoded in the considered fragment. 
Indeed, as we discussed in Section~\ref{sec:validation}, our tool deals with all the type enforcement constructs used in real-word CIL configurations.

Our extension targets well known problems in policy development. 
Moreover, it provides a basis for developing and implementing new high-level languages for SELinux as our semantics completes the existing, informal, and incomplete, CIL documentation.
Our proposal can also be applied to check properties when composing code written in different high-level languages sharing this common intermediate language.

Since the actual SELinux architecture uses CIL as an intermediate language, our tool can also be used to verify properties of configurations written in the current policy language. 
This includes the SELinux reference policy that is part of several Linux distributions, and the Android policy~\cite{AndroidOSP}.

\paragraph{Future work}
There are several exciting directions for future work that aim at fostering the adoption of \lang\ by practitioners.
First, we plan to cover all the features of the CIL language,
even though the type enforcement fragment that we currently support suffices to analyze many real-word configurations.
We will also provide more friendly diagnostics and suggestions for fixing violated requirements.

	We plan to enhance our tool's efficiency by reengineering and optimizing its code, and extending it to fully support requirement refinement.
Also we will address the issues of modular and incremental analysis. 
We consider these aspects critical for the integration of \lang\ in the life-cycle of CIL configurations.
In particular, we aim at supporting the development of tools like IDEs that provide instant feedback to administrators while they are writing their configurations, as is sometimes the case with typed languages.

Finally, we plan to support configurations partly written in the kernel policy language and partly written in CIL, as this is common practice~\cite{aosp}.

\section*{ACKNOWLEDGEMENTS}

We thank the anonymous referees and the shepherd for their insightful comments and detailed suggestions that helped to greatly improve our presentation. 
L. Ceragioli, P. Degano and L. Galletta have been partially supported by the MIUR project PRIN 2017FTXR7S IT MATTERS.

\bibliographystyle{IEEEtran}
\bibliography{biblio}

\appendices
\section{Formalizing CIL}\label{app:CIL}
%

\begin{figure*}
\begin{center}
\small
\begin{tabular}{c}
\prftree[l]
{(N-1)}
{\overline{eval}_{\sigma ; \#}(B) = B'}
{(\sigma, inherit\ B) \rightarrow (\sigma, inherit\ B')}

\quad

\prftree[l]
{(N-2)}
{
	\begin{matrix}
		(\sigma, inherit\ B) \in \Gamma\\
		(B.\rho, r) \in \Gamma\\
	\end{matrix}
}
{add (\sigma.\rho, r)}

\quad

\prftree[l]
{(N-3)}
{\overline{eval}_{\sigma ; \#}(m) = m'}
{(\sigma, call\ m([a])) \rightarrow
(\sigma, call\ m'([a]))}

\\[.4cm]

\prftree[l]
{(N-4)}
{
	\begin{matrix}
		(\sigma, call\ m([a])) \in \Gamma\\
		(m, d) \in \Gamma\\
	\end{matrix}
}
{add (\sigma, d)}

\quad

\prftree[l]
{(N-5a)}
{\begin{matrix}
a\ occurs\ in\ c\\
\overline{eval}_m(a) = a' \neq \bot\\
eval_m(a) = \bot\\
\end{matrix}}
{(m, c) \rightarrow
((m, c\{ a' / a \})}

\quad

\prftree[l]
{(N-5b)}
{\begin{matrix}
(\sigma, call\ \sigma'.n([a])) \in \Gamma\\
(\sigma'.n, c) \in \Gamma\\
(\sigma', macro\ m([x])) \in \Gamma\\
(\lnot\exists m', a' : (\sigma'.n, call\ m'[a']) \in \Gamma)
\end{matrix}}
{add (\sigma, c\{ [a] / [x] \})}

\\[.4cm]

\prftree[l]
{(N-5c)}
{\begin{matrix}
(\sigma, call\ m([a])) \in \Gamma\\
(\lnot\exists m', [a'] : (m, call\ m'[a']) \in \Gamma)
\end{matrix}}
{remove (\sigma, call\ m([a]))}

\qquad

\prftree[l]
{(N-6)}
{
	\begin{matrix}
	a\ occurs\ in\ c\\
	\overline{eval}_{B_\# ; \#}(a) = a'\\
	\end{matrix}
}
{(B_\#, c) \rightarrow
(B_\#, c \{ a' / a \})}

\end{tabular}
\end{center}
\caption{CIL normalization rules.}
\label{fig:CILnormalization}
\end{figure*}

\begin{figure*}
\begin{center}
\small
\begin{tabular}{c}
\prftree[l]
{(S-N1)}
{(B_\#, type\ n) \in \Gamma}
{B_\#.n \in N}

\qquad

\prftree[l]
{(S-N2)}
{(B_\#, typeattribute\ n) \in \Gamma}
{B_\#.n \in N}

\qquad

\prftree[l]
{(S-ta1)}
{(B_\#, type\ n) \in \Gamma}
{(B_\#.n, B_\#.n) \in ta}

\\[.4cm]

\prftree[l]
{(S-ta2)}
{
	\begin{matrix}
	(B_\#, typeattributeset\ a\  (expr))) \in \Gamma\\
	(B_\#', type\ n) \in \Gamma\\
	\semantics{expr}_\Gamma(B_\#'.t)\\
	\end{matrix}
}
{(a, B_\#'.t) \in ta}

\qquad

\prftree[l]
{(S-A1)}
{(B_\#, allow\ t\ t'\ (class\ (perms))) \in \Gamma}
{(t, perms, t') \in A}

\qquad

\prftree[l]
{(S-A2)}
{
	\begin{matrix}
	(t1, perms, t2) \in A\\
	t1' \in ta(t1)\\
	t2' \in ta(t2)\\
	\end{matrix}
}
{(t1', perms, t2') \in A}

\end{tabular}
\end{center}
\caption{CIL semantics.}
\label{fig:CILsemantics}
\end{figure*}

CIL has qualified names, i.e., names prefixed by a path $\rho$ in the nesting of blocks and macros.
$\rho$ is a list of elements separated by dots.
In CIL, qualified names from the global namespace start with a dot ($.$).
We instead use the distinguished symbol $\#$.  
Global qualifications $\sigma$ start from the global namespace $\#$; other qualifications are called relative (e.g., $\#.A.a$ is globally qualified, $A.a$ is a relatively qualified). 
%
%

The syntax of a CIL configuration is as follows, where $[F]$ represents lists of $F$ entities, and $\mathit{CIL}$ is the starting symbol.
\begin{smalign*}
\mathit{CIL} &::= [rule]\\
rule &::= declaration \mid command\\
declaration &::= (block\ n\ CIL) \mid (typeattribute\ n)\\
     & \mid (type\ n) \mid (macro\ n ([x]) (CIL))\\ 
command &::= (allow\ a\ a\ (class\ (perms)))\\
     & \mid (typeattributeset\ a\ (expr))\\
     & \mid (call\ m ([a])) \mid (blockinherit\ B)
\end{smalign*} 
\noindent
Here $n, n', \dots$ are unqualified names (of types, typeattributes, macros or blocks); $x, x', \dots$ are formal parameter names;
$a, a', \dots$ are types and typeattributes (possibly qualified); $m, m', \dots$ are macro (possibly qualified) names; $B, B', \dots$ are block (possibly qualified) names.
Furthermore, we will use $g, g', \dots$ for possibly qualified names of macros or blocks, and $p, p', \dots$ for possibly qualified names in general.
Finally, we use $ B_\#, B_\#', \dots $ to refer to either $\#$ or a block name $B$.

We abstractly represent a CIL configuration as a set of pairs $(\sigma, r)$, where the rule $r$ occurs in the namespace $\sigma$.

Assume as given a set of CIL rules $\Gamma$.
We define the following function $eval_{\sigma}^{k}(p)$ to resolve names, where
$p = \rho.n$ (the qualification $\rho$ is possibly empty, and $n$ is the unqualified name) occurring in the globally qualified block or macro $\sigma$, and where $k \in \{ \texttt{type}, \texttt{typeattribute}, \texttt{block}, \texttt{macro} \}$ 
indicates that we are resolving a type, a typeattribute, etc.
This function returns a fully qualified name 
for $p$, or $\bot$ is resolution is not possible in $\sigma$.
%
\begin{smalign*}
&eval_{\sigma}^{k}(\rho.n) =
\begin{cases}
\rho.n             		      & \text{if}\ \rho = \#.\rho'\\
\sigma.\rho.n           	  & \text{if}\ \rho = \#.\rho'\ \land (\sigma.\rho, k\ n) \in \Gamma\\
\bot			              & \text{otherwise}
\end{cases}
\end{smalign*}
%
Since typeattributes are treated as types in CIL, we abuse notation and simply write $eval_{\sigma}^{type}(a)$ for the function defined as $eval_{\sigma}^{type}(a)$ if its result is different from $\bot$, and as $eval_{\sigma}^{typeattribute}(a)$ otherwise.
Moreover, we omit $k$, assuming that the correct parameter is used.
In CIL, a name that cannot be resolved in the namespace in which it occurs is often resolved in its parent namespace (recursively).
We formalize this as follows
\begin{smalign*}
&\overline{eval}_{\sigma}^{k}(p) =
\begin{cases}
eval_{\sigma}^{k}(p)              & \text{if}\ eval_{\sigma}^{k}(p) \neq \bot\\
\overline{eval}_{\sigma'}^{k}(p)  & \text{if}\ eval_{\sigma}^{k}(p) = \bot\ \land\\
								  & \sigma = \sigma'.n \texttt{ with } \sigma' \neq \#\\
\bot			                  & \text{otherwise}
\end{cases}
\end{smalign*}
Moreover, it is common that a name is resolved in the global namespace if the resolution in the block or macro in which the name is used fails.
We will write $eval_{\sigma;\sigma'}(p)$ for a function that, evaluates $p$ in $\sigma$ unless the result is $\bot$, and evaluates it in $\sigma'$ otherwise.

%

The CIL normalization pipeline consists of the six rewriting rules in \figurename~\ref{fig:CILnormalization}, with 
applicability conditions in the upper part and an action in the lower one.
We denote rules by $r, r', \dots$; declarations by $d, d', \dots$; and commands by $c, c', \dots$.
The conditions predicate on the configuration at hand, and the actions either prescribe: (i) to rewrite a rule in $\Gamma$, as in $(\sigma, r) \rightarrow (\sigma', r')$; (ii) to add a rule in $\Gamma$, as in $add\ (\sigma, r)$; and (iii) to remove a rule from $\Gamma$, as in $remove (\sigma, r)$.
Each phase is iterated until a fixpoint is reached, with the only exception being the fifth phase.
The fifth phase is a sub-pipeline where rule (N-5a) is applied first, then rule (N-5b), and finally (N-5c).
In the fifth phase, each rule and the whole pipeline are applied until a fixpoint is reached.

The rules for CIL's semantics are given
in \figurename~\ref{fig:CILsemantics}, which yield the graph $G = (N, ta, A)$.
Since attribute expressions $expr$ are boolean functions on types and typeattributes, we assume a denotational semantics $\semantics{expr}: N \rightarrow \{\text{true}, \text{false}\}$.

\section{Formalizing IFCIL}\label{app:IFCIL}
%

\subsection{IFL}

\begin{figure*}
\begin{center}
\small
\begin{tabular}{c}
\prftree[l]
{(node)}
{\_}
{n \preceq_n *}

\quad

\prftree[l]
{(op)}
{o \subseteq o'}
{o \preceq_o o'}

\quad

\prftree[l]
{(arrow)}
{\_}
{> \ \ \preceq_a\ +>}

\quad

\prftree[l]
{(o-arrow)}
{o \preceq_o o'}
{w \preceq_a w'}
{(w,o) \preceq_{oa} (w',o')}

\quad

\prftree[l]
{(comp)}
{P_1 \preceq_P P_1'}
{P_2 \preceq_P P_2'}
{P_1 P_2 \preceq_P P_1' P_2'}

\\[.5cm]

\prftree[l]
{(P-1)}
{n_1 \preceq_n n_1'}
{n_2 \preceq_n n_2'}
{oa \preceq_{oa} oa'}
{n_1\ oa\ n_2 \preceq_P n_1'\ oa'\ n_2'}

\qquad

\prftree[l]
{(P-2)}
{n_1 \preceq_n n_1'}
{n_2 \preceq_n n_2'}
{o_1 \preceq_{o} o_1'}
{o_1 \preceq_{o} o_2'}
{o_2 \preceq_{o} o_2'}
{n_1\ +[o_1]>\ *\ [o_2]> n_2 \preceq_P n_1'\ [o_1']>\ *\ +[o_2']>\ n_2'}

\\[.5cm]

\prftree[l]
{(P-3)}
{n_1 \preceq_n n_1'}
{n_2 \preceq_n n_2'}
{o_1 \preceq_{o} o'}
{o_2 \preceq_{o} o'}
{n_1\ +[o_1]>\ *\ +[o_2]> n_2 \preceq_P n_1'\ +[o']>\ n_2'}

\qquad

\prftree[l]
{(P-4)}
{n_1 \preceq_n n_1'}
{n_2 \preceq_n n_2'}
{o_1 \preceq_{o} o_1'}
{o_2 \preceq_{o} o_2'}
{o_2 \preceq_{o} o_1'}
{n_1\ [o_1]>\ *\ +[o_2]> n_2 \preceq_P n_1'\ +[o_1']>\ *\ [o_2']>\ n_2'}

\\[.5cm]

\qquad

\prftree[l]
{($\mathcal{R}$-1)}
{P \preceq_P P'}
{P \preceq P}

\qquad

\prftree[l]
{($\mathcal{R}$-2)}
{P \preceq_P P'}
{\text{\textasciitilde} P' \preceq \text{\textasciitilde} P}

\qquad

\prftree[l]
{($\mathcal{R}$-3)}
{P_1 \preceq_P P_1'}
{P_2 \preceq_P P_2'}
{P_1' : P_2 \preceq P_1 : P_2'}

\end{tabular}
\end{center}
\caption{Definition of IFL requirement refinement and auxiliary relations.}
\label{fig:refinement}
\end{figure*}

IFL requirement refinement is defined by the $\preceq$ operator in \figurename~\ref{fig:refinement}, where reflexivity and transitivity rules are implicitly assumed for every defined relation, and where arrows $>$ or $+>$ are sometimes represented by $w$, and an arrow $w$ labeled with a set of operations $o$ is represented as a pair $(w, o)$ (e.g., $+[\{read\}]>$ is represented as $(+>, \{read\})$).
\begin{lemma}[Kind refinement]\label{thm:kindref}
Let $I$ be an information flow diagram and $\pi$ a path in $I$, if $P \preceq_P P’$ and $\pi \triangleright_I P$ then $\pi \triangleright_I P’$.
\end{lemma}

\begin{proof}
Assume $\pi \triangleright_I P$ and not $\pi \triangleright_I P’$ and proceed by proving for each rule 
in \figurename~\ref{fig:refinement} and for arbirary $\pi$ that if the lemma holds 
on the premises, that it also does on the conclusion.

Consider the rule (comp), and let $i \in \{1,2\}$.
From $\pi_i \triangleright_I P_i$ we have that $\pi_i \triangleright_I P_i'$. 
Then $\pi_1 \pi_2 \triangleright_I P_1' P_2'$ by the definition of $\triangleright$.

Consider the rule (P-1) and proceed by cases on the form of the arrow in $oa$.
If the arrow is $>$ then from $\pi \triangleright_I (n_1 [o]> n_2)$ we know that $\pi = (n_1'', o'', n_2'')$ (since this is the only possible case in the definition of $\triangleright$).
We also know that $n_1'' \in ta(n_1)$ or $n_1 = *$, but since $n_1 \preceq n_1'$ either $n_1 = n_1'$ or $n_1' = *$,
and then $n_1'' \in ta(n_1')$ or $n_1' = *$.
Similarly for $n_2''$.
Finally, from the definition of $\triangleright$ 
we also know that $o \cap o'' \neq \emptyset$, but since $o \preceq o'$, $o \subseteq o'$, thus $o \cap o' \neq \emptyset$.
All the requirements for $\pi \triangleright_I (n_1' [o']> n_2')$ are thus verified.
If the arrow is $+>$, two rules in the definition of $\triangleright$ may be used.
The first one has same condition of the previous case, and the second one 
requires induction on $\pi$ (note that $* \preceq *$).

Consider the rule (P-2) and assume $\pi \triangleright_I n_1\ +[o_1]>\ *\ [o_2]> n_2$.
Then, by the definition of $\triangleright$ we have that $\pi = \pi' (n, o, n_2'')$ with $\pi' \triangleright_I n_1\ +[o_1]>\ *$; $n_2'' \in ta(n_2)$ if $n_2 \neq *$, and $o \cap o_2 \neq \emptyset$.
From $\pi' \triangleright_I n_1\ +[o_1]>\ *$ it follows
 that either $\pi' = (n_1'', o', n')$ or $\pi' = (n_1'', o', n') \pi''$; in both cases $n_1'' \in ta(n_1)$ if $n_1 \neq *$, and $o' \cap o_1 \neq \emptyset$.
If $\pi' = (n_1'', o', n')$ then the thesis holds since $\pi' \triangleright_I n_1'\ [o_1']>\ *$ and $(n, o, n_2'') \triangleright_I * [o_2']> n_2'$ (thus $(n, o, n_2'') \triangleright_I * +[o_2']> n_2'$ by the definition of $\triangleright$).
If $\pi' = (n_1'', o', n') \pi''$, then from $\pi' \triangleright_I n_1\ +[o_1]>\ *$ and all $(n'', o'', n''')$ in $\pi'$, $o'' \cap o_1 \neq \emptyset$ holds.
The thesis follows, since $o_1 \preceq o_2'$, $\pi'' (n, o, n_2'') \triangleright_I * +[o_2']> n_2'$.

Consider the rule (P-3) and assume $\pi \triangleright_I n_1\ +[o_1]>\ *\ +[o_2]> n_2$.
Then split $\pi$ as $\pi' \pi''$ such that $\pi' \triangleright_I n_1\ +[o_1]>\ *$ and $\pi'' \triangleright_I *\ +[o_2]> n_2$.
From $\pi' \triangleright_I n_1\ +[o_1]>\ *$,
the first element of $\pi'$ is some $(n, o, n')$ such that $n \in ta (n_1)$ if $n_1 \neq *$ and thus that $n \in ta (n_1')$ if $n_1' \neq *$. 
Moreover, for every $(n_1'', o_1', n_2'')$ in $\pi'$, $o_1' \cap o_1 \neq \emptyset$.
From $\pi'' \triangleright_*\ +[o_1]>\ n_2$, the last element of $\pi'$ is some $(n'', o'', n''')$ such that $n''' \in ta (n_2)$ if $n_2 \neq *$ and thus $n \in ta (n_1')$ if $n_1' \neq *$. 
Moreover, for every $(n_1''', o_2', n_2''')$ in $\pi''$, $o_2' \cap o_2 \neq \emptyset$.
Finally, since $o_1 \preceq o'$, and $o_2 \preceq o'$, for every $(\_, o''', \_)$ in $\pi$, $o''' \cap o' \neq \emptyset$.

The proof of rule (P-4) is almost the same of (P-2).
\end{proof}

\refinement*
\begin{proof}
We consider the rules in \figurename~\ref{fig:refinement}, and for arbitrary $I$ we show that if the premises are met, then the theorem holds for the conclusion.

Consider the rule ($\mathcal{R}$-1) and assume that $I \models P$, then there exists a path $\pi$ in $I$ such that $\pi \triangleright_I P$.
Since $P \preceq P'$, $\pi \triangleright_I P'$ by Lemma~\ref{thm:kindref}, thus $I \models P'$.

Consider the rule ($\mathcal{R}$-2), and assume that $I \models \text{\textasciitilde}P'$.
We proceed by refutation, assuming $I \not\models \text{\textasciitilde}P$ and showing a contradiction.
By the definition, $I \not\models \text{\textasciitilde}P$ implies $I \models P$ and thus there exists a path $\pi$ in $I$ such that $\pi \triangleright_I P$.
Since $P \preceq P'$, $\pi \triangleright_I P'$ by Lemma~\ref{thm:kindref} and thus $I \models P'$, i.e., $I \not\models \text{\textasciitilde}P'$.

Consider the rule ($\mathcal{R}$-3), and assume that $I \models P_1' : P_2$.
Assume now by refutation that $I \not\models P_1 : P_2'$.
By the definition, $I \not\models P_1 : P_2'$ implies that there exists a path $\pi$ in $I$ such that $\pi \triangleright_I P_1$ and not $\pi \triangleright_I P_2'$.
Since $P_1 \preceq P_1'$, $\pi \triangleright_I P_1'$ by Lemma~\ref{thm:kindref}.
Thus, since $I \models P_1' : P_2$, $\pi \triangleright_I P_2$ must hold, but since $P_2 \preceq P_2'$,
then $\pi \triangleright_I P_2'$ must also hold.
\end{proof}

The (not always defined) \emph{greatest lower bound}, and \emph{least upper bound} of a pair of requirements $\mathcal{R}$ and $\mathcal{R}'$ are represented as $\mathcal{R} \sqcap \mathcal{R}'$, and $\mathcal{R} \sqcup \mathcal{R}'$ respectively.

\subsection{Syntax and semantics of \lang}

\begin{figure*}
\begin{center}
\small
\begin{tabular}{c}
\prftree[l]
{(N-1)}
{\overline{eval}_{\sigma ; \#}(B) = B'}
{(\sigma, inherit\ B\ R) \rightarrow (\sigma, inherit\ B'\ R)}

\quad

\prftree[l]
{(N-2)}
{	\begin{matrix}
	(\sigma, inherit\ B\ R) \in \Gamma\\
	(B.\rho, \overline{r}) \in \Gamma\\
	\end{matrix}}
{add (\sigma.\rho, \overline{r})}

\quad

\prftree[l]
{(N-2')}
{	\begin{matrix}
	(\sigma, inherit\ B\ R) \in \Gamma\\
	(B.\rho, \texttt{;IFL;}\ (l)\ \mathcal{R} \ \texttt{;IFL;}) \in \Gamma\\
	\texttt{;IFL;} (l':\rho.l)\ \mathcal{R}' \texttt{;IFL;} \in R\\
	\end{matrix}}
{add (\sigma.\rho, \texttt{;IFL;} (l')\ \mathcal{R}' \sqcap \mathcal{R} \ \texttt{;IFL;})}

\\[.5cm]

\prftree[l]
{(N-2'')}
{	\begin{matrix}
	(\sigma, inherit\ B\ R) \in \Gamma\\
	(B.\rho, \texttt{;IFL;}\ (l)\ \mathcal{R} \ \texttt{;IFL;}) \in \Gamma\\
	\lnot\exists l', \mathcal{R}' : \texttt{;IFL;} (l':\rho.l)\ \mathcal{R}' \texttt{;IFL;} \in R\\
\end{matrix}}
{add (\sigma.\rho, \texttt{;IFL;} (l')\ \mathcal{R} \ \texttt{;IFL;})}

\quad

\prftree[l]
{(N-3)}
{\overline{eval}_{\sigma ; \#}(m) = m'}
{(\sigma, call\ m([a])\ R) \rightarrow
(\sigma, call\ m'([a])\ R)}

\quad

\prftree[l]
{(N-4)}
{	\begin{matrix}
	(\sigma, call\ m([a])\ R) \in \Gamma\\
	(m, d) \in \Gamma\\
\end{matrix}}
{add (\sigma, d)}

\\[.5cm]

\prftree[l]
{(N-5a)}
{\begin{matrix}
a\ occurs\ in\ c\\
\overline{eval}_m(a) = a' \neq \bot\\
eval_m(a) = \bot\\
\end{matrix}}
{(m, c) \rightarrow
((m, c\{ a' / a \})}

\quad

\prftree[l]
{(N-5b)}
{\begin{matrix}
(\sigma, call\ \sigma'.n([a])\ R) \in \Gamma\\
(\sigma'.n, \overline{c}) \in \Gamma\\
(\sigma', macro\ m([x])) \in \Gamma\\
(\lnot\exists m', a', R' : (\sigma'.n, call\ m'[a']\ R') \in \Gamma)
\end{matrix}}
{add (\sigma, \overline{c}\{ [a] / [x] \})}

\\[.5cm]

\prftree[l]
{(N-5b')}
{\begin{matrix}
(\sigma, call\ \sigma'.n([a])\ R) \in \Gamma\\
(\sigma'.n, \texttt{;IFL;} (l)\ \mathcal{R} \ \texttt{;IFL;}) \in \Gamma\\
(\sigma', macro\ m([x])) \in \Gamma\\
(\lnot\exists m', a', R' : (\sigma'.n, call\ m'[a']\ R') \in \Gamma)\\
{\texttt{;IFL;} (l':l)\ \mathcal{R}' \texttt{;IFL;} \in R}
\end{matrix}}
{add (\sigma, \texttt{;IFL;} (l')\ (\mathcal{R} \sqcap \mathcal{R}')\{ [a] / [x] \} \ \texttt{;IFL;})}

\qquad

\prftree[l]
{(N-5b'')}
{\begin{matrix}
(\sigma, call\ \sigma'.n([a])\ R) \in \Gamma\\
(\sigma'.n, \texttt{;IFL;} (l)\ \mathcal{R} \ \texttt{;IFL;}) \in \Gamma\\
(\sigma', macro\ m([x])) \in \Gamma\\
(\lnot\exists m', a', R' : (\sigma'.n, call\ m'[a']\ R') \in \Gamma)\\
{\lnot\exists l', \mathcal{R}' : \texttt{;IFL;} (l':l)\ \mathcal{R}' \texttt{;IFL;} \in R}
\end{matrix}}
{add (\sigma, \texttt{;IFL;} (l)\ \mathcal{R}\{ [a] / [x] \} \ \texttt{;IFL;})}

\\[.5cm]

\prftree[l]
{(N-5c)}
{\begin{matrix}
(\sigma, call\ m([a])\ R) \in \Gamma\\
(\lnot\exists m', [a'] : (m, call\ m'[a']\ R) \in \Gamma)
\end{matrix}}
{remove (\sigma, call\ m([a])\ R)}

\quad

\prftree[l]
{(N-6)}
{	\begin{matrix}
	a\ occurs\ in\ c\\
	\overline{eval}_{B_\# ; \#}(a) = a'\\
\end{matrix}}
{(B_\#, c) \rightarrow
(B_\#, c \{ a' / a \})}

\end{tabular}
\end{center}
\caption{IFCIL normalization rules.}
\label{fig:IFCILnormalization}
\end{figure*}

The grammar for \lang\ is obtained by updating the rules of $command$ for $call$ and $blockinherit$ as follows.
\begin{smalign*}
command &::= (call\ m ([a])\ [IFL\mathit{refinement}])\\
     & \mid (blockinherit\ B\ [IFL\mathit{refinement}])
\end{smalign*}
Moreover, the following rules are added to the grammar.
\begin{smalign*}
command &::= IFLrequirement\\
IFLrequirement &::=\ \texttt{;IFL;}\ (label)\ \mathcal{R}\ \texttt{;IFL;}\\
IFL\mathit{refinement} &::=\ \texttt{;IFL;}\ (label:label)\ \mathcal{R}\ \texttt{;IFL;}
\end{smalign*}

In the normalization pipeline, the rules are updated as shown in \figurename~\ref{fig:IFCILnormalization}, where $\overline{r}$ and $\overline{c}$ are rules and commands that are not IFL requirements.
The normalization procedure is the same as CIL, where rules (N-i), (N-i'), and (N-i'') are applied together during phase $i$ ($1 \leq i \leq 6$).

The semantics $(G, \mathbb{R})$ of a \lang\ configuration $\Sigma$ is obtained by applying the rules in \figurename~\ref{fig:CILsemantics} for $G$, and the following one for $\mathbb{R}$.
\begin{smalign*}
\prftree[l]
{(S-$\mathbb{R}$)}
{(B_\#, \texttt{;IFL;} (l)\ \mathcal{R} \ \texttt{;IFL;}) \in \Gamma}
{\mathcal{R} \in \mathbb{R}}
\end{smalign*}

\subsection{Verification}

It is convenient to use in the following the grammar of $P$ of Section~\ref{sec:verification}, and to redefine the relation  $\triangleright_I$ of Definition~\ref{def:ifpath-kind} in the following,  trivially equivalent way.
\begin{smalign*}
(\mathtt{n}, o, \mathtt{n'}) \triangleright_I  \mathtt{m}\ \texttt{[}o'\texttt{]>}\ \mathtt{m'} 
          \ \text{ if }\  &(m = *\ \lor\ \mathtt{n}\in \ta(m)) 
          \\ & \> \land (m' = *\ \lor\ \mathtt{n'}\in \ta(m')) 
          \\ & \> \land  o \cap o' \neq \emptyset 
          \\
(\mathtt{n}, o, \mathtt{n'}) \triangleright_I  \mathtt{m}\ \texttt{+[}o'\texttt{]>}\ \mathtt{m'} 
          \ \text{ if }\ & 
          (\mathtt{n}, o, \mathtt{n'}) \triangleright_I  \mathtt{m}\ \texttt{[}o'\texttt{]>}\ \mathtt{m'} 
\\        
(\mathtt{n}, o, \mathtt{n'})\, \pi  \triangleright_I \mathtt{m}\ \texttt{+[}o'\texttt{]>}\ \mathtt{m'} 
          \ \text{ if }\ & (\mathtt{n}, o, \mathtt{n'}) \triangleright_I \mathtt{m}\ \texttt{[}o'\texttt{]>}\ \mathtt{*}\\
          &\> \land\, \pi \triangleright_I \mathtt{*}\ \texttt{+[}o'\texttt{]>}\ \mathtt{m'}  
\end{smalign*}
\begin{smalign*}                
(\mathtt{n}, o, \mathtt{n'})\, \pi  \triangleright_I \mathtt{m}\ \texttt{[}o'\texttt{]>}\ \mathtt{m'}\ P 
\ \text{ if }\ & (\mathtt{n}, o, \mathtt{n'}) \triangleright_I \mathtt{m}\ \texttt{[}o'\texttt{]>}\ \mathtt{*}\\
&\> \land\, \pi \triangleright_I P\\
\pi  \triangleright_I \mathtt{m}\ \texttt{+[}o'\texttt{]>}\ \mathtt{m'}\ P 
          \ \text{ if }\ & \pi  \triangleright_I \mathtt{m}\ \texttt{[}o'\texttt{]>}\ \mathtt{m'}\ P\\
(\mathtt{n}, o, \mathtt{n'})\, \pi  \triangleright_I \mathtt{m}\ \texttt{+[}o'\texttt{]>}\ \mathtt{m'}\ P 
\ \text{ if }\ & (\mathtt{n}, o, \mathtt{n'}) \triangleright_I \mathtt{m}\ \texttt{[}o'\texttt{]>}\ \mathtt{*}
\\ &  \> \land \ \pi \triangleright_I \mathtt{*}\ \texttt{+[}o'\texttt{]>}\ \mathtt{*}\ P
\end{smalign*}
\noindent
Given a path $\pi = (n_0, o_0, n_1) (n_1, o_1, n_2) \dots (n_{m-1}, o_{m-1}, n_m)$ of an information diagram $I$ with $K$ its KTS, let $\ltlencode{\pi}$ be the set of paths $w = n_0, op_0, n_1, op_1, n_2,\dots, n_{m-1}, op_{m-1}, n_m$ in $K$ such that $\forall i : op_i \in o_i$.
Moreover, given a path of $K$ $w = n_0, op_0, n_1, op_1, n_2,\dots, n_{m-1}, op_{m-1}, n_m$, let $\pi_w$ be \mbox{$(n_0, \{op_0\}, n_1) (n_1, \{op_1\}, n_2) \dots (n_{m-1}, \{op_{m-1}\}, n_m)$ of $I$}.

\begin{lemma}\label{thm:paths}
Let $\pi$ be a path in the information diagram $I$ with $K$ its KTS, and let $P$ be an information flow kind.
Then $\pi \triangleright_I P$ if and only if $w \models_l \ltlencode{P}$ for some $w \in \ltlencode{\pi}$.
\end{lemma}

\begin{proof}
We proceed by induction on $P$.
For $P$ of the form $\mathtt{m}\ \texttt{[}o'\texttt{]>}\ \mathtt{m'}$ or $\mathtt{m}\ \texttt{+[}o'\texttt{]>}\ \mathtt{m'}$ the properties trivially hold.
For the last two cases of $P$ we proceed as follows.

\textbf{Case} $P = \mathtt{m}\ \texttt{[}o'\texttt{]>}\ \mathtt{m'}\ P'$.
Assume $\pi \triangleright_I P$.
By the definition of $\triangleright_I$, $\pi = (n, o, n')\ \pi'$, with $o \cap o' \neq \emptyset$.
We also know by the induction hypothesis that there exists 
$w' \in \ltlencode{\pi'}$ such that $w' \models_l \ltlencode{P'}$.
Since $o \cap o' \neq \emptyset$, there exists $op \in o \cap o'$.
Take $w = n, op, w'$.
Then $w \in \ltlencode{\pi}$ holds by definition, and also $w \models_l \ltlencode{P}$.

Assume $w \models_l \ltlencode{P}$.
Since $\ltlencode{P} = m \land \bigvee_{op \in o'} (op)\ \land X(\ltlencode{P'})$, $w = n, op, w'$ with $n \in ta(m)$ and $op \in o'$ trivially follows from the semantics of LTL.
Then $\pi = (n, o, n')\pi'$ with $op \in o$ and $w' \in \ltlencode{\pi'}$.
We also know by the induction hypothesis that $\pi' \triangleright_I P'$, then $\pi \triangleright_I P$.

\textbf{Case} $P = \mathtt{m}\ \texttt{+[}o'\texttt{]>}\ \mathtt{m'}\ P'$.
Assume $\pi \triangleright_I P$.
We separately consider the two cases of the definition of $\triangleright$ that apply to $P$.
If $\pi  \triangleright_I \mathtt{m}\ \texttt{[}o'\texttt{]>}\ \mathtt{m'}\ P'$, then we are in the previous case.
Otherwise, $\pi = (n, o, n')\ \pi'$ with $(n, o, n') \triangleright_I m \texttt{[}o'\texttt{]>} *$ and $\pi' \triangleright_I * \texttt{[}o'\texttt{]>} *\ P'$.
By induction hypothesis, $w' \in \ltlencode{\pi'}$ exists such that $w' \models_l \bigvee_{op \in o'} (op)\ \land X(\bigvee_{op \in o'} (op)\ U\ \ltlencode{P'})$.
Thus, $w = n,op,w'$ exists such that $op \in o$, and it holds that $w \in \ltlencode{\pi}$.
The thesis follows from the following stronger statement that trivially holds by construction: $w \models_l m \land \bigvee_{op \in o'} (op)\ \land X(\bigvee_{op \in o'} (op) \land X (\bigvee_{op \in o'} (op)\ U\ \ltlencode{P'}))$.
%

Assume $w \models_l \ltlencode{P}$.
Then either $w \models_l \bigvee_{op \in o'} (op)\ \land X(\ltlencode{P'})$ or $w \models_l \bigvee_{op \in o'} (op)\ \land X(\bigvee_{op \in o'} (op)\ \land X(\bigvee_{op \in o'} (op)\ U\ \ltlencode{P'}))$.
If $w \models_l \bigvee_{op \in o'} (op)\ \land X(\ltlencode{P'})$, then we are in the previous case.
Consider now the second case, and let $w$ be $n, op, w'$.
By definition, $\pi = (n, o, n') \pi'$, with $op \in o$ and $w' \in \ltlencode{\pi'}$.
By construction, $(n, o, n') \triangleright_I m \texttt{[}o'\texttt{]>} *$, and by induction hypothesis, $\pi' \triangleright_I * \texttt{+[}o'\texttt{]>} * P$.
\end{proof}

\begin{lemma}\label{thm:paths2}
	Let $w$ be a path in the KTS $K$ of the information diagram $I$, and let $P$ be an information flow kind.
	Then $w \models_l \ltlencode{P}$ if and only if $\pi_w \triangleright_I P$.
\end{lemma}
\begin{proof}
	Trivially derives from Lemma~\ref{thm:paths} since $\ltlencode{\pi_w} = \{ w \}$.
\end{proof}


\correctness*

\begin{proof}
We proceed by cases on $\mathcal{R}$.

\textbf{Case} $\mathcal{R} = P$.
Assume $I \models P$.
Then there exists $\pi$ such that $\pi \triangleright_I P$, and by Lemma~\ref{thm:paths} $K \vdash P$.
Conversely, assume $K \vdash P$, thus there exists $w$ such that $w \models_l \ltlencode{P}$, and by Lemma~\ref{thm:paths2} $I \models P$.

\textbf{Case} $\mathcal{R} = \text{\textasciitilde}P$.
Assume $I \models \text{\textasciitilde}P$.
Then there exists no $\pi$ such that $\pi \triangleright_I P$.
Assume by refutation that $K \not\vdash P$.
Then $w \models_l \ltlencode{P}$ for some $w$, and by Lemma~\ref{thm:paths2} $\pi_w \triangleright_I P$.
Contradiction.

Conversely, assume $K \vdash \text{\textasciitilde}P$, thus there exists no $w$ such that $w \models_l \ltlencode{P}$.
Assume by refutation that $I \not\models \text{\textasciitilde}P$.
Then there exists $\pi$ such that $\pi \triangleright_I P$ and, by Lemma~\ref{thm:paths2} there exists $w$ such that $w \models_l \ltlencode{P}$.
Contradiction.

\textbf{Case} $\mathcal{R} = P:P'$.
Assume $I \models P:P'$.
Then there exists no $\pi$ such that $\pi \triangleright_I P$ and $\pi \not\triangleright_I P$.
Assume by refutation that $K \not\vdash P:P'$.
Then there exists $w$ such that $w \models_l \ltlencode{P}$ and $w \not\models_l \ltlencode{P'}$, and by Lemma~\ref{thm:paths2} $\pi_w \triangleright_I P$ and $\pi_w \not\triangleright_I P$.
Contradiction.

Conversely, assume $K \vdash P:P'$, thus there exists no $w$ such that $w \models_l \ltlencode{P}$ and $w \not\models_l \ltlencode{P'}$.
Assume by refutation that $I \not\models P:P'$.
Then there exists $\pi$ such that $\pi \triangleright_I P$ and $\pi \not\triangleright_I P'$.
Thus, by Lemma~\ref{thm:paths}, there exists $w$ such that $w \models_l \ltlencode{P}$ and $w \not\models_l \ltlencode{P'}$.
Contradiction.
\end{proof}

In the following we write $\Bar{W}$ and $\Bar{W_\iota}$ for infinite paths in $K$ and $K_\iota$; 
$\Dot{W}$ and $\Dot{W_\iota}$ for finite paths in $K$ and $K_\iota$; and $W_\iota$ for $\Bar{W_\iota} \cup \Dot{W_\iota}$ (note that $W = \Bar{W} \cup \Dot{W}$).
With a small abuse of notation, we write $K_\iota \models_l \phi$ if and only if $\forall w \in \Bar{W_\iota}. w \models_l \phi$.
%
\begin{definition}
The encoding of flow kinds for $K_\iota$ is as follows.
\begin{smalign*}
	&\ltlencode{\ttype{n [o]> n}'}_\iota = \ttype{n} \land \bigvee_{\ttype{op} \in \ttype{o}} (\ttype{op}) \land X(\ttype{n}' \land X(\iota) )\\
	&\ltlencode{\ttype{n +[o]> n}'}_\iota = \ttype{n} \land \bigvee_{\ttype{op} \in \ttype{o}} (\ttype{op}) \land X(\bigvee_{\ttype{op} \in \ttype{o}} (\ttype{op})\ U\ (\ttype{n}' \land X(\iota)))\\
	&\ltlencode{(\ttype{n [o]> n}') P}_\iota = \ttype{n} \land \bigvee_{\ttype{op} \in \ttype{o}} (\ttype{op}) \land X(\ltlencode{P}_\iota)\\
	&\ltlencode{(\ttype{n +[o]> n}') P}_\iota = \ttype{n} \land \bigvee_{\ttype{op} \in \ttype{o}} (\ttype{op}) \land X(
	\bigvee_{\ttype{op} \in \ttype{o}} (\ttype{op})\ U\ \ltlencode{P}_\iota)
\end{smalign*}
The satisfaction relation $\vdash_\iota$ is defined as follows.
\begin{smalign*}
	K_\iota \vdash_\iota P &\text{ iff } K_\iota \not\models_l \lnot\ltlencode{P}_\iota\\
	K_\iota \vdash_\iota \ttype{\textasciitilde}P &\text{ iff } K_\iota \models_l \lnot\ltlencode{P}_\iota\\
	K_\iota \vdash_\iota P \ttype{:} P' &\text{ iff } K_\iota \models_l \lnot\ltlencode{P}_\iota \lor \ltlencode{P'}_\iota
\end{smalign*}
\end{definition}

\begin{lemma}\label{thm:finiota}
If $w \models_l \ltlencode{P}$ then $w \in \Dot{W}$ and $w \iota^{\omega} \models_l \ltlencode{P}_\iota$ with $w \iota^{\omega} \in \Bar{W_\iota}$.
\end{lemma}
\begin{proof}
Trivial.
\end{proof}

\begin{lemma}\label{thm:iotafin}
If $w \models_l \ltlencode{P}_\iota$ then 
there exists a unique ${\dot w} \in \Dot{W}$ such that ${\dot w} \models_l \ltlencode{P}$ and $w = {\dot w} \iota^\omega$.
\end{lemma}
\begin{proof}
Trivial.
\end{proof}

\begin{corollary}\label{thm:MC}
	Let $\Sigma$ be an IFCIL configuration with requirements $\mathbb{R}$, let $I$ be its information flow diagram, and let $K$ be the KTS of $\Sigma$. Then
	\vspace{-1mm}
	\begin{align*}
		K_\iota \vdash \mathbb{R}\ \text{  if and only if  }\  K \vdash \mathbb{R}\ \text{  if and only if  }\ I \models \mathbb{R}.
	\end{align*}
\end{corollary}

\begin{proof}
	It suffices to prove the following.
	\begin{align}
	K \models_l \lnot \ltlencode{P} &\text{ iff } K_\iota \models_l \lnot \ltlencode{P}_\iota\label{case1}\\
	K \models_l \lnot \ltlencode{P} \lor \ltlencode{P'} &\text{ iff } K_\iota \models_l \lnot \ltlencode{P}_\iota \lor \ltlencode{P'}_\iota\label{case2}
	\end{align}
	\textbf{Proof of (\ref{case1}).}
	Assume by refutation that $K \models_l \lnot \ltlencode{P}$ and $K_\iota \not\models_l \lnot \ltlencode{P}_\iota$.
	From $K_\iota \not\models_l \lnot \ltlencode{P}_\iota$ we know that there exists $w' \in \Bar{W_\iota}$ such that $w' \models_l \ltlencode{P}_\iota$.
	By Lemma~\ref{thm:iotafin}, there exists $w'' \in W$ such that $w'' \models_l \ltlencode{P}$.
	But from $K \models_l \lnot \ltlencode{P}$ we know that $\forall w \in W. w \not\models_l \ltlencode{P}$. Contradiction.
	
	Assume now by refutation that $K \not\models_l \lnot \ltlencode{P}$ and $K_\iota \models_l \lnot \ltlencode{P}_\iota$.
	From $K \not\models_l \lnot \ltlencode{P}$ we know that there exists $w' \in W$ such that $w \models_l \ltlencode{P}$.
	By Lemma~\ref{thm:finiota}, $w' \iota^\omega \in \Bar{W_\iota}$ is such that $w' \iota^\omega \models_l \ltlencode{P}_\iota$.
	But from $K_\iota \models_l \lnot \ltlencode{P}_\iota$ we know that $\forall w \in \Bar{W_\iota}$ it is $w \not\models_l \ltlencode{P}$. Contradiction.
	
	\textbf{Proof of (\ref{case2}).}
	Assume by refutation that $K \models_l \lnot \ltlencode{P} \lor \ltlencode{P'}$ and $K_\iota \not\models_l \lnot \ltlencode{P}_\iota \lor \ltlencode{P'}_\iota$.
	From $K_\iota \not\models_l \lnot \ltlencode{P}_\iota \lor \ltlencode{P'}_\iota$ we know that there exists $w' \in \Bar{W_\iota}$ such that $w' \models_l \ltlencode{P}_\iota$ and $w' \not\models_l \ltlencode{P'}_\iota$.
	Then, by Lemma~\ref{thm:iotafin}, $w' = w'' \iota^\omega$ with $w'' \in \Dot{W}$ and $w'' \models_l \ltlencode{P}$.
	Clearly, $w'' \not\models_l \ltlencode{P'}$, otherwise we would have that $w' \models_l \ltlencode{P'}_\iota$ by Lemma~\ref{thm:finiota}.
	But from $K \models_l \lnot \ltlencode{P} \lor \ltlencode{P'}$ we know that $\forall w \in W$ it is $w \not\models_l \ltlencode{P} \lor w \models_l \ltlencode{P'}$. Contradiction.
	
	Assume now by refutation that $K \not\models_l \lnot \ltlencode{P} \lor \ltlencode{P'}$ and $K_\iota \models_l \lnot \ltlencode{P}_\iota \lor \ltlencode{P'}_\iota$.
	From $K \not\models_l \lnot \ltlencode{P} \lor \ltlencode{P'}$ we know that there exists $w' \in W$ such that $w' \models_l \ltlencode{P}$ and $w' \not\models_l \ltlencode{P'}$.
	Thus, by Lemma~\ref{thm:finiota}, $w' \iota^\omega \models_l \ltlencode{P}_\iota$ with $w' \iota^\omega \in \Bar{W_\iota}$.
	Clearly, $w' \iota^\omega \not\models_l \ltlencode{P'}_{\iota}$, otherwise we would have that $w' \models_l \ltlencode{P'}$ by Lemma~\ref{thm:finiota}.
	But from $K_\iota \models_l \lnot \ltlencode{P}_\iota \lor \ltlencode{P'}_\iota$ we know that $\forall w \in \Bar{W_\iota}$ it is  $w \not\models_l \ltlencode{P}_\iota \lor w \models_l \ltlencode{P'}_\iota$. Contradiction.
\end{proof}


\end{document}